\documentclass[11pt]{article}
\usepackage{theorem}
\usepackage{latexsym}
\usepackage{graphicx}
\usepackage{amssymb}
\usepackage{amsmath}
\usepackage{MnSymbol}
\reversemarginpar
\theoremheaderfont{\bf}
\theorembodyfont{\sl}
\newtheorem{theo}{Theorem}[section]
{\theorembodyfont{\rm} \newtheorem{defi}[theo]{Definition}}
{\theorembodyfont{\rm} }
{\theorembodyfont{\rm} \newtheorem{exa}[theo]{Example}}
{\theorembodyfont{\rm} \newtheorem{rem}[theo]{Remark}}
\newtheorem{prop}[theo]{Proposition}

{\theorembodyfont{\rm} }
\newtheorem{lemma}[theo]{Lemma}
{\theorembodyfont{\rm}}
{\theorembodyfont{\rm}\newtheorem{notation}[theo]{Notation}}
\newenvironment{proof}{{\sc Proof:}}{\mbox{}\hfill$\Box$\par}
\numberwithin{equation}{section}
\newcommand{\eqnref}[1]{~\mbox{$(${\rm \ref{#1}}$)$}}

\newcommand{\junk}[1]{}

\newcommand{\TS}{\textstyle}

\newcommand{\F}{{\mathbb F}}

\newcommand{\cC}{{\mathcal C}}
\newcommand{\cP}{{\mathcal P}}
\newcommand{\cR}{{\mathcal R}}
\newcommand{\cL}{{\mathcal L}}
\newcommand{\cS}{{\mathcal S}}

\newcommand{\cI}{{\mathcal I}}
\newcommand{\cA}{{\mathcal A}}
\newcommand{\cB}{{\mathcal B}}
\newcommand{\cT}{{\mathcal T}}

\newcommand{\cK}{{\mathcal K}}

\newcommand{\cV}{{\mathcal V}}

\newcommand{\rk}{\mbox{${\rm rk}$\,}}

\newcommand{\inner}[1]{\mbox{$\langle{#1}\rangle$}}

\newcommand{\edge}[1]{\mbox{$\longrightarrow$}\hspace{-1.3em}%
\raisebox{1.3ex}{${\scriptscriptstyle{#1}}$}\hspace{.8em}}
\newcommand{\im}{\mbox{\rm im}\,}
\newcommand{\row}{\,\mbox{${\rm row}$}}

\newcommand{\T}{\mbox{$^{\sf T}$}}
\renewcommand{\mod}{\,{\rm mod}\,}

\newcounter{alp}
\newcounter{ara}
\newcounter{rom}
\newenvironment{romanlist}{\begin{list}{(\roman{rom})\hfill}{\usecounter{rom}
     \topsep0ex \labelwidth.7cm \leftmargin.7cm \labelsep0cm
     \rightmargin0cm \parsep0ex \itemsep.6ex
     \partopsep1.6ex}}{\end{list}}
\newenvironment{alphalist}{\begin{list}{(\alph{alp})\hfill}{\usecounter{alp}
     \topsep0ex \labelwidth.7cm \leftmargin.7cm \labelsep0cm
     \rightmargin0cm \parsep0ex \itemsep.6ex
     \partopsep1.6ex}}{\end{list}}
\newenvironment{arabiclist}{\begin{list}{(\arabic{ara})\hfill}{\usecounter{ara}
     \topsep0ex \labelwidth.7cm \leftmargin.7cm \labelsep0cm
     \rightmargin0cm \parsep0ex \itemsep.6ex
     \partopsep1.6ex}}{\end{list}}

\title{Characteristic Generators and Dualization for Tail-Biting Trellises}
\date{January~3, 2011}
\author{Heide Gluesing-Luerssen\thanks{The author was partially supported by National Science Foundation
        grant \#DMS-0908379}\ \; and Elizabeth~A.~Weaver
       \\
       University of Kentucky\\
       Department of Mathematics\\
       715 Patterson Office Tower\\
       Lexington, KY 40506-0027, USA;
       \\
       heide.gl@uky.edu, eaweaver1s@uky.edu
       }

\begin{document}
\maketitle
\noindent{\bf Abstract:}
This paper focuses on dualizing tail-biting trellises, particularly KV-trellises.
These trellises are based on characteristic generators, as introduced by Koetter/Vardy (2003), and
may be regarded as a natural generalization of minimal conventional trellises, even though
they are not necessarily minimal.
Two dualization techniques will be investigated:
the local dualization, introduced by Forney (2001) for general normal graphs, and
a linear algebra based dualization tailored to the specific class of tail-biting BCJR-trellises,
introduced by Nori/Shankar (2006).
It turns out that, in general, the BCJR-dual is a subtrellis of the local dual, while for
KV-trellises these two coincide.
Furthermore, making use of both the BCJR-construction and the local dualization, it will be shown
that for each complete set of characteristic generators of a code there exists a complete set of characteristic
generators of the dual code such that their resulting KV-trellises are dual to each other if paired suitably.
This proves a stronger version of a conjecture formulated by Koetter/Vardy.

\vspace*{.3cm}

\vspace*{.3cm}

\section{Introduction}\label{S-Intro}
\setcounter{equation}{0}

It is well known that for a given linear block code, a tail-biting trellis may be smaller than the
minimal conventional trellis with respect to any of the various notions of complexity of a trellis; see
the discussion in \cite[Sec.~III]{KoVa03}.
Since iterative decoding on tail-biting trellises is well understood (as opposed to decoding on more
general graphs with cycles), this has led to an increased interest in the construction of minimal
tail-biting trellises; see also~\cite{KoVa03,LiSh00,NoSh06,ShBe00,YaQi07,ZhOh07}.
A major breakthrough in this direction has been obtained by Koetter/Vardy~\cite{KoVa03}.
They showed that for each $k$-dimensional linear block code of length~$n$ with full support there exists
a list of~$n$ characteristic generators, each endowed with a span interval, such that every minimal tail-biting
trellis of the code is structurally isomorphic to a product of the elementary trellises of~$k$ linearly
independent characteristic generators (where structurally isomorphic means trellis isomorphic, but disregarding the edge labels);
see \cite[Thm.~5.5]{KoVa03} and the adjustments in \cite[Prop.~III.14, Thm.~III.15]{GLW10}.
Here minimality may refer to any of the standard complexity notions for tail-biting trellises discussed in
\cite[Thm.~5.6]{KoVa03}.
Moreover, the same results show that for each minimal trellis there is a suitable choice of~$k$ linearly independent
characteristic generators such that the given trellis is isomorphic to the resulting product trellis.
This shows a major difference from conventional trellises: while the minimal conventional trellis of a
block code is unique up to trellis isomorphism, this is not the case in the tail-biting situation~-- again, this
refers to any minimality notion.
Not only may there exist minimal trellises with incomparable state complexity profiles
or edge complexity profiles, but even if both these profiles coincide for two minimal trellises,
the trellises may not be isomorphic (but they are structurally isomorphic in this situation due to
\cite[Prop.~III.14]{GLW10}).

We will use the term KV-trellises for product trellises based on~$k$ linearly independent
characteristic generators in the sense described above.
From the construction performed in~\cite{KoVa03}, it follows that characteristic generators may be regarded as
a generalization of MSGM's or trellis-oriented generator matrices in the realm of conventional trellises,
see \cite[Def.~6.2]{McE96} or \cite[Sec.~IV]{KschSo95}, or shortest bases in the sense of~\cite{Fo09};
see also Lemma~\ref{L-ShortestGen} and Remark~\ref{R-greedy} in the next section.
As a consequence, KV-trellises, though not minimal in general, form a natural generalization of minimal
conventional trellises, and indeed, these trellises have much nicer properties than more general tail-biting trellises.
In the paper~\cite{GLW10}, two trellis constructions based on generators with
span intervals have been investigated: the product construction and the BCJR-construction, introduced
by Nori/Shankar~\cite{NoSh06}.
It has been shown that, in general, a BCJR-trellis is smaller than the corresponding product trellis.
In fact, the latter can be merged to the former by taking suitable quotients of its state spaces~\cite[Thm.~IV.9]{GLW10}.
For KV-trellises, however, these two constructions are isomorphic~\cite[Thm.~IV.11]{GLW10}.
As a consequence, KV-trellises are non-mergeable.
Another demonstration of the distinctiveness of KV-trellises will be given in this paper.
It will be shown that KV-trellises behave significantly nicer under dualization than more general trellises.

We will investigate two dualization techniques for tail-biting trellises.
Both lead to trellises representing the dual code.
The first construction is a specialization of the local dualization introduced by Forney in~\cite{Fo01} for
general normal graphs.
It amounts to dualizing the transition spaces along with a sign inverter; see also~\cite{Fo11} for a different approach
based on graphical models.
This dualization has been generalized to factor graphs in~\cite{MaoKsch05} and recently been
recast in the framework of Valiant transforms~\cite{AlMao10}.
Even though the local dualization is a very elegant and convenient construction, for tail-biting trellises it may
lead to dual trellises with some undesirable properties; see Example~\ref{E-Localdual1}.
The second construction is a simple linear algebra based dualization for BCJR-trellises as introduced by Nori/Shankar~\cite{NoSh06}.
In Section~\ref{S-Dual} we will see that the BCJR-dual is a subtrellis of the local dual.
For KV-trellises, however, these duals coincide.
Furthermore, as we will show in Section~\ref{S-DualProc}, the dual of a KV-trellis is a KV-trellis again and thus
shares all their nice properties.

More specifically, in Section~\ref{S-DualProc} we will prove that for each set of~$n$ characteristic generators of a given
code~$\cC\subseteq\F^n$, there exists a set of~$n$ characteristic generators of the dual code~$\cC^{\perp}$
such that the dual of each KV-trellis of~$\cC$ based on the chosen generators is (isomorphic to) a
KV-trellis of~$\cC^{\perp}$ based on the dual generators.
We will construct the list of dual generators explicitly and also show the direct link between the~$k$ linearly independent
characteristic generators for~$\cC$ and the $n-k$ dual characteristic generators that give rise
to the dual KV-trellis.
In fact, this link is easily described because Koetter/Vardy have shown~\cite[Thm.~5.12]{KoVa03} that the characteristic
span list of~$\cC^{\perp}$ is simply obtained by reversing the characteristic spans of~$\cC$.
The construction of the dual list of characteristic generators is performed as follows.
One starts with the BCJR-trellis of~$\cC$ based on the entire list of chosen characteristic generators.
This trellis gives rise to~$n$ subtrellises obtained by omitting exactly one characteristic generator.
They still represent the code~$\cC$.
The main result of the procedure tells us that the local dual of each such subtrellis contains a cycle that gives rise to
a dual characteristic generator whose span is the reversal of the span that had been omitted.
All of this leads to the desired~$n$ dual characteristic generators.
With the principles of local dualization as well as Koetter/Vardy's result of reversed characteristic spans in mind, this
approach of constructing dual characteristic generators is quite natural.
The details however, carried out in Section~\ref{S-DualProc}, become rather technical due
to certain linear independence conditions that need to be verified.

\medskip

Let us close the introduction with introducing the basic notions needed for this paper.
Throughout, a {\sl tail-biting trellis\/} $T=(V,E)$ of depth~$n$ over the finite field~$\F$  is a directed edge-labeled graph
with the property that the vertex set~$V$ partitions into $n$ disjoint sets $V=V_0\cup V_1\cup\ldots\cup V_{n-1}$
such that every edge in~$T$ that starts in~$V_i$ ends in $V_{i+1\mod n}$.
The edges are labeled with field elements from~$\F$.
Notice that we compute modulo~$n$ on the {\sl time axis\/}~$\cI:=\{0,\ldots,n-1\}$.
Referring to the fact that a trellis is a state space realization of the code regarded as a
dynamical system (behavior) on~$\cI$, we call~$V_i$ the {\sl state space\/} of the trellis at time~$i$,
and its elements are the states at that time.
The edge set~$E$ decomposes into $E=\bigcup_{i=0}^{n-1}E_i,$ where $E_i$ is the set of edges starting
in~$V_i$ and ending in~$V_{i+1\mod n}$.
Its elements reflect the present-state to next-state transitions, and therefore the
edge sets~$E_i$ will be called {\sl transition spaces}.
We  identify the elements of~$E_i$ (the edges) with the triples consisting of starting state, label, and ending state.
Thus, the transition spaces are given by
$E_i=\{(v,a,\hat{v})\mid \text{there exists an edge } v\edge{\;a}\,\hat{v} \text{ where }v\in V_i,\,\hat{v}\in V_{i+1},\,a\in\F\}
 \subseteq V_i\times\F\times V_{i+1}$
for $i\in\cI$.
These spaces have also been called trellis sections~\cite{CFV99,FT93} or local constraints~\cite{Fo01}.

A {\sl cycle\/} in~$T$ is a closed path of length~$n$ in the trellis.
We always assume that the cycles start and end (at the same state) in~$V_0$.
If~$|V_0|=1$, the trellis is called {\sl conventional}.
We call the trellis {\sl reduced\/} if every state and every edge appear in at least one cycle.
The trellis is called {\sl biproper\/} if any two edges starting at the same vertex or ending at the same
vertex are labeled distinctly.

The trellis~$T$ is {\sl linear\/} if each state space~$V_i$ is a vector space over~$\F$ and
the {\sl label code\/}
\begin{equation}\label{e-ST}
   \cS(T)\!=\!\{(v_0,\ldots,v_{n-1},c)\in V_0\times\ldots\times V_{n-1}\times\F^n\!\mid\!
   v_0\edge{c_0}\!v_1\edge{c_1}\!\ldots\!v_{n-1}\edge{\!\!\!\!c_{n-1}}\!\!\!v_0\text{ is a cycle in~$T$}\}
\end{equation}
is a subspace of $V_0\times\ldots\times V_{n-1}\times\F^n$ and if the transition spaces~$E_i$ are linear
subspaces of $V_i\times\F\times V_{i+1}$.
If the trellis is reduced, then the linearity of~$E_i$ follows from the linearity of $\cS(T)$.
We say that~$T$ {\sl represents the code\/}~$\cC\subseteq\F^n$ if~$\cC$ equals its {\sl edge-label code}, that is,
$\cC=\{(c_0,\ldots,c_{n-1})\in\F^n\mid
\text{there exists a cycle }v_0\edge{c_0}\!v_1\edge{c_1}\!\ldots\!\edge{\!\!\!\!c_{n-1}}\!\!v_0\text{ in~$T$}\}$.

Note that if~$T$ is linear, then the represented block code is linear.
We will only deal with linear block codes and linear trellis representations.
The trellis~$T$ is called {\sl one-to-one\/} if distinct cycles in~$T$ have distinct edge-label sequences.
The {\sl state complexity profile (SCP)\/} and  {\sl edge complexity profile (ECP)\/} of a linear trellis~$T=(V,E)$ are
defined as $\text{SCP}(T):=(s_0,\ldots,s_{n-1})$, where $s_i=\dim V_i$, and
$\text{ECP}(T):=(e_0,\ldots,e_{n-1})$, where $e_i=\dim E_i$.
Throughout this paper, the notion of {\sl minimality\/} for tail-biting trellises refers to any of the orderings
discussed by Koetter/Vardy in \cite[Sec.~III]{KoVa03}.
While for conventional trellises all these minimality notions coincide, this is not the case for tail-biting trellises.
In~\cite[Thm.~5.5, Thm.~5.6]{KoVa03} and \cite[Thm.~III.15]{GLW10} it has been shown that a minimal trellis (with respect to any of those orderings)
is a KV-trellis in the sense of our Definition~\ref{D-CharMat}.
In this paper we will be concerned with KV-trellises, and a specific notion of minimality will not be needed.

Linear trellises $T=(V,E)$ and $T'=(V',E')$ are called {\sl isomorphic\/} if there exists a bijection
$\phi:V\longrightarrow V'$ such that $\phi(V_i)=V'_i$ and $\phi|_{\TS V_i}:V_i\longrightarrow V'_i$ is an
isomorphism for all $i\in\cI$ and $(v,\,a,\,w)\in E_i$ if and only if $(\phi(v)\,\,a,\,\phi(w))\in E'_i$.
Obviously, isomorphic trellises represent the same code.


Finally, we fix the following notation pertaining to the code under consideration and its representation.
Throughout, let
\begin{equation}\label{e-Cdata}
   \cC=\im G=\ker H\T\subseteq\F^n\text{ be a $k$-dimensional code with support } \cI=\{0,\ldots,n-1\},
\end{equation}
where the latter means that for each $j\in\cI$ there exists a codeword~$(c_0,\ldots,c_{n-1})\in\cC$ such that
$c_j\not=0$.
Here, $\im M:=\{\alpha M\!\mid\! \alpha\in\F^m\}$ and $\ker M:=\{\alpha\in\F^m\!\mid\! \alpha M=0\}$ denote the
row space and left kernel of the matrix~$M\in\F^{m\times n}$, respectively.
We assume $G\in\F^{r\times n}$,  hence $\rk G=k\leq r$, and will explicitly state when $r=k$ and thus~$G$
is a full row rank encoder matrix.
Throughout, $H\in\F^{(n-k)\times n}$  is a full row rank parity check matrix.
Furthermore, we fix the notation
\begin{equation}\label{e-Gdata}
    G=(g_{lj})_{l=1,\ldots,r\ \;\atop j=0,\ldots,n-1}
     =\begin{pmatrix}G_0^{\sf T}&\ldots&G_{n-1}^{\sf T}\end{pmatrix}\in\F^{r\times n}\;\text{ and }\;
    H\T=\begin{pmatrix}H_0\\\vdots\\ H_{n-1}\end{pmatrix}\in\F^{n\times(n-k)}.
\end{equation}
Hence $G_j^{\sf T}\in\F^r$ and $H_j^{\sf T}\in\F^{n-k}$ are the columns of~$G$ and~$H$, respectively.
Finally, in order to avoid extreme cases, we will also assume that~$\cC^{\perp}$ has support~$\cI$.
As for the matrices~$G$ and~$H$ above, we will use the notation $M_j^{\sf T}$ for the $j$-th column
of the matrix~$M$ and we will employ the (Maple) notation $\row(M,l)$ for the $l$-th row of~$M$.

\section{KV-Trellises and the BCJR-Construction}\label{S-Prelim}
\setcounter{equation}{0}
We will begin by briefly recalling the main results about products of (tail-biting) elementary trellises.
Thereafter we turn to KV-trellises, the product trellises obtained by choosing $k$ linearly independent characteristic
generators of the code as introduced by Koetter/Vardy in~\cite{KoVa03}.
Finally, we will discuss the BCJR-construction of trellises and recall some results from~\cite{GLW10} pertaining to the
relation between the product- and the BCJR-construction.

Due to the cyclic structure of the time axis~$\cI$, the following interval
notation has proven to be very convenient.
For $a,\,b\in\cI$ we define
$[a,\,b]:=\{a, a+1,\ldots,b\}$ if $a\leq b$ and
$[a,\,b]:=\{a,a+1,\ldots,n-1,0,1,\ldots,b\}$ if $a> b$.
Moreover, we set $(a,\,b]:=[a,b]\backslash\{a\}$.
We call the intervals $(a,\,b]$ and $[a,\,b]$ {\sl conventional\/} if $a\leq b$ and {\sl circular\/} otherwise.
Notice that $(a,a]=\emptyset$.
It is easy to see that $\cI\,\backslash\,(a,\,b]=(b,\,a]$ for all $a\not=b$.
Hence the complement of a nonempty conventional interval is circular and vice versa.
The following notion from \cite[p.~2089]{KoVa03} will be crucial for this paper.

\begin{defi}
\label{D-vectorspan}
For a vector $c=(c_0,\ldots,c_{n-1})\in\F^n\backslash\{0\}$ we call any half-open interval $(a,b]$ a {\sl span\/}
      of~$c$ if $c_a\not=0\not=c_b$ and if the closed interval $[a,b]$ contains the support of~$c$.
\end{defi}
Excluding the starting point~$a$ from the span does not seem to be intuitive but will be very convenient for our purposes.
\footnote{It would be more accurate to distinguish between a time axis for the symbols and a time axis for the
states, see~\cite{Fo09}.
Then the span $(a,b]$ of a vector is its active state interval
(that is, it is the interval of nonzero states of its corresponding cycle in the elementary trellis),
while its active symbol interval is given by $[a,b]$.}

Let~$(a,b]$ be a span of the nonzero vector $c=(c_0,\ldots,c_{n-1})\in\F^n$.
The {\sl elementary trellis for the pair\/} $(c,(a,b])$ is defined as $T_{c,(a,b]}:=(V,E)$ with
the state spaces and transition spaces given by $V_j=\im (\mu_j)\subseteq\F$ and
$E_j=\im(\mu_j,c_j,\mu_{j+1})\subseteq V_j\times\F\times V_{j+1}$, respectively, where $\mu_j=1$ for $j\in(a,b]$
and $\mu_j=0$ otherwise.
Thus~$V_j=\F$ for $j\in(a,b]$ and $V_j=\{0\}$ otherwise.
The trellis  $T_{c,(a,b]}$ is linear, reduced, biproper,  one-to-one, and represents the 1-dimensional code
in~$\F^n$ generated by~$c$.
The trellis is conventional if and only if $(a,b]$ is a conventional span.
Obviously, the SCP and ECP are given by $(s_0,\ldots,s_{n-1})$ and $(e_0,\ldots,e_{n-1})$, respectively,
where $s_j=1$ if $j\in(a,b]$ and $s_j=0$ if $j\not\in(a,b]$, while $e_j=1$
if $j\in[a,b]$ and $e_j=0$ if $j\not\in[a,b]$.

For the following notion of product trellises recall the definition and basic properties of trellis products
$T_1\times T_2$; see, for instance,~\cite[p.~2089]{KoVa03} and \cite[Prop.~III.4]{GLW10}).

\begin{defi}\label{D-Prodtrellis}
Let $\cC=\im G$, where~$G\in\F^{r\times n}$ has no zero rows.
Denote the rows of~$G$ by $g_1,\ldots,g_r\in\F^n$ and let
$\cS:=[(a_l,b_l],\,l=1,\ldots,r]$ be a {\sl span list\/} for~$G$, that is,~$(a_l,b_l]$ is a span
(conventional or circular) for the row~$g_l,\,l=1,\ldots,r$.
The {\sl product trellis\/} $T_{G,\cS}$ is defined as the trellis $T_{g_1,(a_1,b_1]}\times\ldots\times T_{g_r,(a_r,b_r]}$.
In other words, the state and transition spaces of $T_{G,\cS}$ are given by $V_j=\im M_j$ and
$E_j=\im (M_j,G_j^{\sf T},M_{j+1})$, where, as before,~$G_j^{\sf T}$ denotes the $j$-th column of~$G$ and
\[
   M_j=\begin{pmatrix}\mu^1_{j}& & \\ &\ddots& \\& &\mu^r_{j}\end{pmatrix}
        \in\F^{r\times r}, \text{ where }
   \mu^l_{j}=\left\{\begin{array}{ll}1,&\text{if }j\in (a_l,b_l],\\
                                           0,&\text{if }j\not\in(a_l,b_l].
                   \end{array}\right.
\]
\end{defi}
Let us briefly comment on the subtle differences in the naming of such trellises in the literature.
First of all, due to Definition~\ref{D-vectorspan}, we only consider spans such that
the vector is nonzero at the endpoints of that span (as opposed to general intervals containing the support of the vector).
This way, we immediately exclude certain (but not all) mergeable trellises; see also ~\cite[Lemma~4.3]{KoVa03}.
This makes our definition more restrictive than the one in~\cite{KoVa03}.
Second, in the paper~\cite{NoSh06}, product trellises of the type above (and
where~$G$ has full row rank) are called KV-trellises, and the matrix~$G$, along with its span list, is called a KV-product matrix.
In the present paper, we will reserve the name KV-trellises for a particular type of product trellis that has been introduced
by Koetter/Vardy~\cite{KoVa03}; see Definition~\ref{D-CharMat}.

The following properties of product trellises are easy to see.
Later on we will only deal with product trellises where the starting points (resp., ending points) of the spans are distinct,
and therefore we restrict ourselves to this case in~(d) below.
One could easily give the formulas for the more general case.

\begin{prop}\label{P-formulas}
Let the data be as in the previous definition and put $T:=T_{G,\cS}$.
Then
\begin{alphalist}
\item $T$ is a linear and reduced trellis.
\item $T$ is one-to-one if and only if $\rk G=r$.
\item $T$ is biproper if and only if $a_1,\ldots,a_r$ are distinct and $b_1,\ldots,b_r$ are distinct.
\item Let $a_1,\ldots,a_r$ be distinct and $b_1,\ldots,b_r$ be distinct.
      Then the SCP and ECP of~$T$ are given by $(s_0,\ldots,s_{n-1})$ and $(e_0,\ldots,e_{n-1})$, respectively, where
      $s_j=|\{l=1,\ldots,r\mid j\in(a_l,b_l]\}|$ and
      \[
        e_j=\left\{\begin{array}{ll}s_j+1,&\text{if }j\in\{a_1,\ldots,a_r\}\\s_j,&\text{if }j\not\in\{a_1,\ldots,a_r\}\end{array}\right\}
        =\left\{\begin{array}{ll}s_{j+1}+1,&\text{if }j\in\{b_1,\ldots,b_r\}\\s_{j+1},&\text{if }j\not\in\{b_1,\ldots,b_r\}.\end{array}\right\}
      \]
\end{alphalist}
\end{prop}

\begin{proof}
(a),~(b) as well as the formulas in~(d) follow easily from the properties of elementary trellises as well as
those of trellis products; see also \cite[Prop.~III.4]{GLW10}.
\\
(c) The implication ``$\Rightarrow$ '' follows from the proof of \cite[Cor.~4.5]{KoVa03}.
The converse has been shown in \cite[Thm.~III.6]{GLW10}.
\end{proof}

We now turn to a particular class of product trellises, introduced by Koetter/Vardy~\cite{KoVa03}.
A main result of~\cite{KoVa03} -- and a major breakthrough in the study of minimal tail-biting trellises --
is the construction of a list of characteristic generators, collected in a characteristic matrix,  from
which all minimal trellises can be derived.
In~\cite{KoVa03}, this matrix is defined as the outcome of a particular procedure.
We will follow the presentation in~\cite{GLW10} and define the characteristic matrix in terms of its
relevant properties.
A justification of this approach versus the one in~\cite{KoVa03} has been given in~\cite[Sec.~III]{GLW10}

\begin{defi}\label{D-CharMat}
Let $\cC\subseteq\F^n$ be as in\eqnref{e-Cdata}.
A {\sl characteristic pair\/} of~$\cC$ is defined to be a pair $(X,\cT)$, where
\begin{equation}\label{e-XY}
  X=\begin{pmatrix}x_1\\ \vdots\\ x_n\end{pmatrix}\in\F^{n\times n}\text{ and }
  \cT=\big[(a_l,b_l],\,l=1,\ldots,n\big]
\end{equation}
have the following properties.
\begin{romanlist}
\item $\im X=\cC$, that is, $\{x_1,\ldots,x_n\}$ forms a generating set of~$\cC$.
\item $(a_l,\,b_l]$ is a span of~$x_l$ for $l=1,\ldots,n$.
\item $a_1,\ldots,a_n$ are distinct and $b_1,\ldots,b_n$ are distinct.
\item For all $j\in\cI$, there exist exactly $n-k$ row indices, $l_1,\ldots,l_{n-k}$, such that
      $j\in(a_{l_i},\,b_{l_i}]$ for $i=1,\ldots,n-k$.
\end{romanlist}
We call~$X$ a {\sl characteristic matrix\/} of~$\cC$ and~$\cT$ the {\sl characteristic span list}.
The rows of~$X$ are also called {\sl characteristic generators}.
A trellis $T_{G,\cS}$ is called a {\sl KV$_{(X,\cT)}$-trellis\/} of~$\cC$ if $G\in\F^{k\times n}$
consists of~$k$ distinct linearly independent rows of~$X$ and $\cS$ consists of the corresponding~$k$ spans in~$\cT$.
A trellis is called a {\sl KV-trellis\/} of~$\cC$ if it is a KV$_{(X,\cT)}$-trellis for some
characteristic pair $(X,\cT)$ of~$\cC$.
\end{defi}

The relevance of characteristic pairs and KV-trellises becomes apparent in part~(b) below:
the class of KV-trellises contains all minimal trellises.
As a consequence, one may restrict the study of tail-biting trellises to the important class of KV-trellises.
For the following results it is crucial that~$\cC$ has support~$\cI$.
\begin{theo}[\mbox{\cite[Sec.~V]{KoVa03}, \cite[Thm.~III.15]{GLW10}}]\label{T-KVresults}
\begin{alphalist}
\item The code~$\cC$ has a characteristic pair, and the characteristic span list is, up to ordering,
      uniquely determined by~$\cC$.
\item For every minimal trellis~$T$ of~$\cC$ there exists a characteristic pair $(X,\cT)$ such that~$T$ is a KV$_{(X,\cT)}$-trellis.
\end{alphalist}
\end{theo}

The following property shows that characteristic spans  are ``shortest spans'' for each given starting point.
This property not only makes it algorithmically easy to find characteristic generators for a given code, see Remark~\ref{R-greedy}
below, but will also be crucial later on in order to derive strong properties for KV-trellises.

\begin{lemma}\label{L-ShortestGen}
Let $c\in\cC$ be a nonzero codeword with (conventional or circular) span
$(a,\,b]$.
Then $(a,\,\hat{b}]\subseteq(a,\,b]$, where $(a,\,\hat{b}]$ is the unique characteristic span starting at~$a$.
\end{lemma}
\begin{proof}
Notice that by Def.~\ref{D-CharMat}(iii), there does indeed exist a characteristic span $(a,\hat{b}]$ starting at~$a$.
Let us first consider the case where $a=0$.
By~(iv) of the same definition there exist exactly~$k$ characteristic spans not containing~$0$, thus~$k$ conventional spans.
One of these spans is $(0,\hat{b}]$.
Definition~\ref{D-CharMat}(iii) tells us that generators with these spans form an MSGM of~$\cC$ in the sense of
\cite[Def.~6.2, Thm.~6.11]{McE96}.
As a consequence, it has minimal span length.
This implies that $(0,\hat{b}]\subseteq(0,b]$, for otherwise we could
replace the generator with span $(0,\hat{b}]$ by~$c$ and obtain a generator matrix with shorter span length.
\\
If $a\not=0$, we may use the cyclic left shift~$\sigma^a$ by~$a$ units in~$\F^n$.
From Definition~\ref{D-CharMat} we obtain immediately that if
$\cT=[(a_1,b_1],\ldots,(a_n,b_n]]$ is the characteristic span list of~$\cC$ with generators $x_1,\ldots,x_n$, then
$[(a_1-a,b_1-a],\ldots,(a_n-a,b_n-a]]$ is the characteristic span list of the code~$\sigma^a(\cC)$ with generators
$\sigma^a(x_1),\ldots,\sigma^a(x_n)$; see also Remark~\ref{R-shift} further down.
Since~$\sigma^a(c)\in\sigma^a(\cC)$ has span $(0,b-a]$, we may now make use of the first case.
This leads to the desired result.
\end{proof}

\begin{rem}\label{R-greedy}
The above result allows us to set up a greedy algorithm for finding the characteristic spans of a given code (with support~$\cI$).
Namely, for $a=0,\ldots,n-1$, let $b_a\in\cI$ be such that $(a,b_a]$ is the shortest span among all possible spans starting at time~$a$
of the  (nonzero) codewords in~$\cC$.
Then the resulting list $[(a,b_a],\,a=0,\ldots,n-1]$ is the characteristic span list of~$\cC$.
The same can be done by screening the spans by their ending points.
Of course, the greedy algorithm will also produce, at the same time, a list of characteristic generators, that is,
a characteristic matrix.
This generalizes the ``shortest basis approach'' described by Forney in \cite{Fo09} for very general (conventional) realizations
to the tail-biting case.
\end{rem}

The following construction of tail-biting trellises has been introduced by Nori/Shankar in~\cite[Sec.~III]{NoSh06} and
studied in detail in~\cite[Sec.~IV]{GLW10}.
We will generalize this slightly by also allowing generator matrices that do not have full rank.
For~(b) recall that $\row(D,l)$ denotes the $l$-th row of the matrix~$D$.

\begin{defi}\label{D-NS}
Let the code~$\cC$ and the matrices~$G\in\F^{r\times n},\,H\in\F^{(n-k)\times n}$ be as in\eqnref{e-Cdata} and\eqnref{e-Gdata}.
\begin{alphalist}
\item Let $D\in\F^{r\times(n-k)}$ be any matrix.
      For $i\in\cI$ define the matrices
      \begin{equation}\label{e-NMat}
        N_0=D\text{ and }N_i=N_{i-1}+G_{i-1}^{\sf T}H_{i-1}\text{ for }i>0.
      \end{equation}
      Then $N_n=N_0$.
      We define~$T_{(G,H,D)}$ to be the trellis with state spaces $V_i:=\im N_i\subseteq\F^{n-k}$ and transition spaces
      $E_i=\im (N_i,G_i^{\sf T},N_{i+1})$.
      It is easy to see that~$T_{(G,H,D)}$ is a linear, reduced, and biproper trellis representing the code~$\cC$.
      We call~$D$ the displacement matrix for the trellis~$T_{(G,H,D)}$.
\item Let $\cS:=[(a_l,b_l],l=1,\ldots,r]$ be a span list of~$G$.
      Then the trellis $T_{(G,H,\cS)}$ is defined as $T_{(G,H,D)}$, where
      \begin{equation}\label{e-DMat}
         D
         \in\F^{r\times(n-k)}\text{ is such that }\row(D,l)=\sum_{j=a_l}^{n-1}g_{lj} H_j\text{ for }j=1,\ldots,r.
      \end{equation}
      The trellis $T_{(G,H,\cS)}$ is called a {\sl (tail-biting) BCJR-trellis\/} of~$\cC$.
\end{alphalist}
\end{defi}
Notice that if $D=0$, then the trellis $T_{(G,H,D)}$ is conventional.
If in addition, $r=k=\rk G$, then $T_{(G,H,D)}$ is, in fact, the classical conventional
BCJR-trellis of~$\cC$, hence minimal, see \cite[Sec.~IV]{McE96} and~\cite{BCJR74}.
In general, the rows of the displacement matrix~$D$ may be interpreted as the (circular) past of the generators in~$G$.
As always in state space realizations, this past is captured in the state at time~$0$ through which the associated
trajectory (cycle) passes.
Consequently, the $l$-th rows of the matrices~$N_j$ form the sequence of states through which the cycle induced by the $l$-th
generator passes.
In this sense, the BCJR-trellis is, just like the product trellis, based on~$r$ individual generators  with spans.
The state spaces arise as the sum of the state spaces of each generator (while for product trellises they are the direct sum of
the state spaces).
As a consequence, BCJR-trellises are not always one-to-one, even if the generator matrix~$G$ has full row rank.

The following results have been proven in~\cite{GLW10} for the case where $G\in\F^{k\times n}$ has rank~$k$.
One can easily verify that the same proofs apply to $G\in\F^{r\times n}$ with $r>k=\rk G$.
For the notion of mergeability, we refer to \cite[Sec.~II.B]{KoVa03} or \cite[Sec.~II]{GLW10}.

\begin{theo}[\mbox{\cite[Cor.~IV.7, Thm.~IV.9, Rem.~IV.13, Thm.~IV.11]{GLW10}}]\label{T-BCJR}
Let~$G,\,H$, and the span list~$\cS$ be as in Definition~\ref{D-NS}.
\begin{alphalist}
\item The BCJR-trellis $T_{(G,H,\cS)}$ is non-mergeable.
\item If the product trellis $T_{G,\cS}$ is non-mergeable, then $T_{G,\cS}$ is isomorphic to $T_{(G,H,\cS)}$.
\item If the product trellis $T_{G,\cS}$  has the same SCP as the BCJR-trellis
      $T_{(G,H,\cS)}$, then these two trellises are isomorphic.
      In this case, the starting points of the spans in~$\cS$ are distinct and so are the ending points.
\end{alphalist}
\end{theo}
It is worth noting that every product trellis $T_{G,\cS}$  can be merged to the corresponding
BCJR-trellis $T_{(G,H,\cS)}$, see \cite[Thm.~IV.9]{GLW10}.

\begin{exa}\label{E-selfdualBCJR}
\begin{alphalist}
\item This example appeared already in \cite[Ex.~IV.12]{GLW10}.
      It will be revisited later again when discussing dualization techniques.
      Consider the code $\cC=\im G=\ker H\T\subseteq\F_2^5$, where
       \[
           G=\begin{pmatrix}0&1&1&1&0\\1&0&0&1&0\\0&1&1&0&1\end{pmatrix},\quad
           H=\begin{pmatrix}1&0&1&1&1\\0&1&1&0&0\end{pmatrix}.
       \]
       Then $\cS:=[(1,3],\, (3,0],\,(2,1]]$ is a span list for~$G$.
       Making use of the greedy algorithm in Remark~\ref{R-greedy} and checking all~$8$ codewords, it is easy to
       see that $(1,3]$ and $(2,1]$ are characteristic
       spans of~$\cC$, whereas the codeword $(1,0,0,0,1)\in\cC$ with span $(4,0]$ shows that
       $(3,0]$ is not a characteristic span.
       The BCJR-trellis $T:=T_{(G,H,\cS)}$ has state and transition spaces given by
       $V_j=\im N_j$ and $E_j=\im(N_j,G_j^{\sf T},N_{j+1})$, where~$N_j$ are given in the matrix
       \begin{equation}\label{e-Niexample}
       \begin{split}
          &(N_0|G_0^{\sf T}|N_1|G_1^{\sf T}|N_2|G_2^{\sf T}|N_3|G_3^{\sf T}|N_4|G_4^{\sf T}|N_0)=\mbox{}\hspace*{10em}\\
          &\mbox{}\hspace*{6em}\left(\!\!\begin{array}{cc|c|cc|c|cc|c|cc|c|cc|c|cc}
            0\!&\!0\!&0&  \!0\!&\!0\!&1& \!0\!&\!1\!&1& \!1\!&\!0\!&1& \!0\!&\!0\!&\!0\!&0\!&\!0\\
            1\!&\!0\!&1&  \!0\!&\!0\!&0& \!0\!&\!0\!&0& \!0\!&\!0\!&1& \!1\!&\!0\!&\!0\!&1\!&\!0\\
            0\!&\!1\!&0&  \!0\!&\!1\!&1& \!0\!&\!0\!&1& \!1\!&\!1\!&0& \!1\!&\!1\!&\!1\!&0\!&\!1\end{array}\!\!\right).
       \end{split}
       \end{equation}
       In this matrix we stagger~$N_j$ with the columns of~$G$ in order to easily read off the transition spaces~$E_j$.
       The trellis~$T$ is shown in the figure below (here and in all other trellises we will denote edges with label one as
       solid lines and those with zero label as dashed lines).
       It is straightforward to see that the corresponding product trellis $T_{G,\cS}$ has the same SCP as~$T$ and therefore
       is isomorphic to~$T$ according to Theorem~\ref{T-BCJR}.
       As a consequence,~$T$ is one-to-one, which we can see also directly from the trellis.
       Notice that by the above,~$T$ is not a KV-trellis.
       \\[1ex]
       \mbox{}\hspace*{4cm}\includegraphics[height=3.5cm]{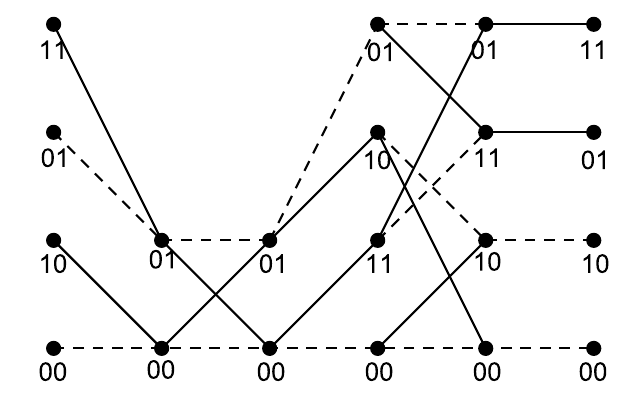}
       \\
       \mbox{}\hspace*{5.5cm}{\footnotesize [Trellis $T$]}
\item Consider the self-dual code $\cC\subseteq\F_2^4$ with $\cC=\im G=\ker G\T$, where
      \[
        G=\begin{pmatrix}1&1&1&1\\0&1&1&0\end{pmatrix}.
      \]
      A characteristic pair is given by
      \begin{equation}\label{e-charmatSD}
        X=\begin{pmatrix}1&0&0&1\\0&1&1&0\\0&1&1&0\\1&1&1&1\end{pmatrix},\
        \cT=\big[(3,0],(2,1],(1,2],(0,3]\big].
      \end{equation}
      This can easily be seen by applying Remark~\ref{R-greedy} to the three nonzero codewords in~$\cC$.
      Let us consider the BCJR-trellis $T_{(X,G,\cT)}$.
      The matrix
      \[
        S=(N_0|X_0^{\sf T}|N_1|X_1^{\sf T}|N_2|X_2^{\sf T}|N_3|X_3^{\sf T}|N_0)=\left(\!\!\begin{array}{cc|c|cc|c|cc|c|cc|c|cc}
            1&0&1&0&0&0&0&0&0&0&0&1&1&0\\1&1&0&1&1&1&0&0&1&1&1&0&1&1\\
            0&0&0&0&0&1&1&1&1&0&0&0&0&0\\0&0&1&1&0&1&0&1&1&1&0&1&0&0\end{array}\!\!\right)
      \]
      lists the state space matrices~$N_j$ staggered with the columns of~$X$.
      We will not display this trellis as it is not a very useful trellis by itself.
      For instance, it is not one-to-one.
      However, we can easily read off the state space matrices and transition spaces for all
      KV$_{(X,\cT)}$-trellises by simply taking the submatrix of~$S$ consisting of any two rows for which the corresponding
      two rows in~$X$ are linearly independent.
      This results in five KV$_{(X,\cT)}$-trellises.
      It is an immediate consequence of Proposition~\ref{P-formulas}(d), that no two of these trellises have the same SCP and ECP.
      In particular, they are pairwise non-isomorphic.
      The trellises $T_{(X,G,\cT)}$, based on an entire characteristic matrix, will be used in Section~\ref{S-DualProc} as the
      starting point of the dualization procedure.
\item Let $\cC=\im G=\ker H\T\subseteq\F_3^4$, where
       \[
           G=\begin{pmatrix}1&2&0&0\\0&0&1&1\end{pmatrix},\ H=\begin{pmatrix}1&1&0&0\\0&0&1&2\end{pmatrix}.
       \]
       The pair $(X,\cT)$, where
       \[
           X=\begin{pmatrix}1&2&0&0\\2&1&0&0\\0&0&1&1\\1&2&1&1\end{pmatrix},\ \cT=\big[(0,1],\,(1,0],\,(2,3],\,(3,2]\big],
       \]
       is a characteristic pair of~$\cC$. The matrix~$X$ is normalized, that is, each characteristic generator has coordinate~$1$ at the
       starting point of its span.
       It can easily be seen that~$\cC$ has~$9$ different normalized characteristic matrices with span list~$\cT$.
\end{alphalist}
\end{exa}

The product construction as well as the BCJR-construction behave nicely under the cyclic shift.
This is described in the following remark, of which we will make frequent use.
Recall that we compute with indices modulo~$n$, which, of course, also applies to the span lists.

\begin{rem}\label{R-shift}
Denote the rows of~$G$ by $g_1,\ldots,g_r\in\F^n$ and let the span list~$\cS$ be as in Definition~\ref{D-NS}(b).
Let~$\sigma$ denote the cyclic left shift on~$\F^n$ and let $G^*\in\F^{r\times n}$ be the matrix consisting
of the shifted rows $\sigma(g_l),\,l=1,\ldots,r$.
Then $\cS^*=[(a_l-1,\, b_l-1],l=1,\ldots,r]$ forms a span list for~$G^*$.
The state and transition spaces of the product trellis $T_{G^*,\cS^*}$ are given by
$V_i^*=\im M_{i+1}$ and $E_i^*=\im(M_{i+1},G_{i+1}^{\sf T},M_{i+2})$, where~$M_i$ is as in Definition~\ref{D-Prodtrellis}.
Similarly, if $(X,\cT)$ as in\eqnref{e-XY} is a characteristic pair for~$\cC$, then
$\big(X^*,\cT^*\big)$ is a characteristic pair of~$\sigma(\cC)$, where
\[
   X^*:=\begin{pmatrix}\sigma(x_1)\\ \vdots\\ \sigma(x_n)\end{pmatrix},\
   \cT^*:=\big[(a_l-1,\, b_l-1],\,l=1,\ldots,n\big].
\]
Finally, if~$D$ is as in\eqnref{e-DMat} and~$N_i$ are the state space matrices of $T_{(G,H,\cS)}$ as in\eqnref{e-NMat},
then the state space matrices for the BCJR-trellis $T_{(G^*,H^*,\cS^*)}$ are given by $N_i^*=N_{i+1}$ for $i\in\cI$.
\end{rem}

In this paper, we will mainly consider BCJR-trellises based on characteristic generators.
The following theorem will be crucial later.

\begin{theo}\label{T-KVBCJR}
Let $(X,\cT)$ be a characteristic pair of~$\cC$.
Let $G\in\F^{r\times n}$ be a selection of~$r$ rows of~$X$ and $\cS:=[(a_l,b_l],l=1,\ldots,r]$ be the
corresponding selection of characteristic spans.
Thus,~$(a_l,b_l]$ is the span of the $l$-th row of~$G$.
Consider the BCJR-trellis $T_{(G,H,\cS)}$, and let $N_j\in\F^{r\times(n-k)},\,j\in\cI$, be its state space matrices.
Then for every $j\in\cI$
\begin{arabiclist}
\item $\row(N_j,l)=0$ for all~$l$ such that $j\not\in(a_l,b_l]$,
\item the set $\{\row(N_j,l)\mid l\text{ such that }j\in(a_l,b_l]\}$ is linearly independent.
\end{arabiclist}
As a consequence,~$T_{(G,H,\cS)}$ is isomorphic to the corresponding product trellis~$T_{G,\cS}$.
Therefore it makes sense to call $T_{(G,H,\cS)}$ a KV-trellis in the case where $r=k=\rk G$.
\end{theo}
\begin{proof}
The statement in~(1) has been proven in \cite[Prop.~IV.6]{GLW10} for BCJR-trellises
$T_{(G,H,\cS)}$, where $G\in\F^{k\times n}$ has rank~$k$.
It can easily be seen that the same proof applies to the general case for~$G$.
As for the second statement, first consider $j=0$.
Then $j\in(a_l,b_l]$ is equivalent to $(a_l,b_l]$ being a circular span.
In the proof of~\cite[Thm.~IV.11]{GLW10} it has been shown that the set
$\{\row(N_0,l)\mid (a_l,b_l]\text{ circular}\}$ is linearly independent
(these rows appear in the matrix~$Z\in\F^{(n-k)\times(n-k)}$ in~\cite[Eq.~(IV.5)]{GLW10}, which is
non-singular due to the same proof in~\cite{GLW10}).
This proves~(2) for $j=0$.
Applying a cyclic left shift by~$j$ steps (see Remark~\ref{R-shift}), we obtain the
desired result for arbitrary~$N_j$, and this completes the proof of~(2).
\\
Finally, the previous results tell us that $\rk N_j=|\{l\mid j\in(a_l,b_l]\}|$.
Furthermore, due to Definition~\ref{D-CharMat}, the starting points of distinct characteristic spans are distinct,
and the same is true for the ending points.
Therefore we may apply Proposition~\ref{P-formulas}(d) and conclude that
$T_{(G,H,\cS)}$ and~$T_{G,\cS}$ have the same SCP.
Thus, these trellises are isomorphic due to Theorem~\ref{T-BCJR}(c).
\end{proof}

We close this section with the following technical results pertaining to the BCJR-presenta\-tion of KV-trellises,
which will be needed later on.

\begin{lemma}\label{L-HN}
Let $T=T_{(G,H,\cS)}$ be a KV-trellis of~$\cC$, that is,
$G=(g_{lj})\in\F^{k\times n}$ has rank~$k$ and its span list~$\cS=[(a_l,b_l]\mid l=1,\ldots,k]$ consists of
characteristic spans of~$\cC$.
Let~$N_j,\,j\in\cI$, be the state space matrices of~$T$.
Then
\begin{romanlist}
\item If $j=b_l$ for some~$l$, then $H_j=-g_{lj}^{-1}\row(N_j,l)$
\item If $j\not\in\{b_1,\ldots,b_k\}$, then $H_j\not\in\im N_j$.
\end{romanlist}
\end{lemma}
As the proof will show, part~(i) is true for general BCJR-trellises, and only~(ii) needs the particular
properties of characteristic spans.

\begin{proof}
Using a cyclic shift, see Remark~\ref{R-shift}, we may assume without loss of generality that
$j=0$.
\\
(i) Let $0=b_l$.
Then\eqnref{e-DMat} implies $\row(N_0,l)=\sum_{i=a_l}^{n-1} g_{li}H_i=-g_{l0}H_0$, where the last identity
follows from the identity $GH\T=0$ along with the fact that
$(a_l,b_l]=(a_l,0]$ is the span of the $l$-th row $(g_{l0},\ldots,g_{l,n-1})$ of~$G$.
This establishes~(i).
\\
(ii)
Let $0\not\in\{b_1,\ldots,b_k\}$ and assume $H_0=\beta N_0$ for some $\beta\in\F^k$.
Using the definition of~$N_0$ in\eqnref{e-DMat}, this becomes
$H_0=\sum_{l=1}^k\beta_l\sum_{j=a_l}^{n-1}g_{lj}H_j$.
Notice that $H_0\not=0$, due to our general assumption that the dual code $\cC^{\perp}$ has support~$\cI$.
Therefore with the aid of Theorem~\ref{T-KVBCJR}(1), the above may be written as
$H_0=\sum_{l\in\cL}\beta_l\sum_{j=a_l}^{n-1}g_{lj}H_j$, where
$\cL:=\{l\mid 0\in(a_l,b_l],\,\beta_l\not=0\}$ and $\cL\not=\emptyset$.
Let $s\in\cL$ be such that $a_s=\min\{a_l\mid l\in\cL\}$.
Then $a_s>0$ because the condition $0\in(a_l,b_l]$ implies that the span is circular.
Define the vectors
\[
  \hat{g}_l=(\hat{g}_{l0},\ldots,\hat{g}_{l,n-1}),\text{ where }
  \hat{g}_{lj}=\left\{\begin{array}{ll}g_{lj},&\text{if }j\geq a_l\\ 0,&\text{if }j<a_l.\end{array}\right.
\]
Then we obtain $H_0=\sum_{l\in\cL}\beta_l\hat{g}_lH\T$.
As a consequence, $c:=\sum_{l\in\cL}\beta_l\hat{g}_l-e_0\in\ker H\T=\cC$, where $e_0\in\F^n$ is
the first standard basis vector.
Now the definition of~$\hat{g}_l$ shows that the codeword~$c$ has span $(a_s,0]$, and thus
Lemma~\ref{L-ShortestGen} implies $(a_s,b_s]\subseteq(a_s,0]$.
Since $(a_s,b_s]$ is circular, this in turn yields $b_s=0$, contradicting our assumption
that~$0\not\in\{b_1,\ldots,b_k\}$.
This proves~(ii).
\end{proof}

\begin{theo}\label{T-bigcapkerNi}
Let~$T_{(G,H,\cS)}$ and~$N_j$ be as in Lemma~\ref{L-HN}. Then $\bigcap_{j=0}^{n-1}\im N_j=\{0\}$.
\end{theo}
\begin{proof}
Let $w\in\bigcap_{j=0}^{n-1}\im N_j$, say $w=\alpha_j N_j$ for some
$\alpha_j=(\alpha_{j,1},\ldots,\alpha_{j,k})\in\F^k,\,j\in\cI$.
We will first show that there exists some common $\alpha\in\F^k$ such that $w=\alpha N_j$ for all $j\in\cI$.
\\
Note that $0=\alpha_{j+1} N_{j+1}-\alpha_j N_j=\alpha_{j+1}(N_j+G_j^{\sf T}H_j)-\alpha_j N_j$, hence
$\alpha_{j+1}G_j^{\sf T}H_j=(\alpha_{j}-\alpha_{j+1})N_j$.
Observe that $\alpha_{j+1}G_j^{\sf T}\in\F$, thus a scalar.
\\
i) If $\alpha_{j+1}G_j^{\sf T}=0$, then $\alpha_{j}-\alpha_{j+1}\in\ker N_j$, and Theorem~\ref{T-KVBCJR}
implies $\alpha_{j,l}=\alpha_{j+1,l}$ for all $l$ such that $j\in(a_l,b_l]$.
\\
ii)~If $\alpha_{j+1}G_j^{\sf T}\not=0$, then the above yields $H_j=(\alpha_{j+1}G_j^{\sf T})^{-1}(\alpha_{j}-\alpha_{j+1})N_j$.
In this case Lemma~\ref{L-HN} implies that $j=b_m$ for some $m\in\{1,\ldots,k\}$.
Furthermore, part~(i) of that Lemma along with Theorem~\ref{T-KVBCJR}(1) and~(2)
shows that
$\alpha_{j,l}=\alpha_{j+1,l}$ for all~$l\not=m$ such that $j\in(a_l,b_l]$.
\\
But then~i) and~ii) together yield
\begin{equation}\label{e-alphaj}
   \alpha_{j,l}=\alpha_{b_l,l} \text{ for all $l$ such that } j\in(a_l,b_l].
\end{equation}
Define now $\alpha:=(\alpha_{b_1,1},\ldots,\alpha_{b_k,k})\in\F^k$.
Using Theorem~\ref{T-KVBCJR}(1) once more, along with\eqnref{e-alphaj}, it is straightforward to see that
$\alpha N_j=\alpha_j N_j=w$ for all $j\in\cI$.
\\
Now we have $0=w-w=\alpha(N_{j+1}-N_j)=\alpha G_j^{\sf T}H_j$ for all $j\in\cI$.
Since $H_j\not=0$ for all $j\in\cI$, this leads to $\alpha G_j^{\sf T}=0$ for all~$j\in\cI$ and hence $\alpha G=0$.
As a consequence, the full row rank of~$G$ implies $\alpha=0$, and we arrived at $w=0$, as desired.
\end{proof}

It is worth noting that the last result is not true for general BCJR-trellises $T_{(G,H,\cS)}$,
even if~$G\in\F^{k\times n}$ has rank~$k$ and the trellis is isomorphic to the corresponding product trellis.
An example is given in Example~\ref{E-selfdualBCJR}(a).
The trellis~$T$ displayed in that graph satisfies $T_{(G,H,\cS)}\cong T_{G,\cS}$
and~$G\in\F_2^{3\times 5}$ has rank~$3$, but $(0,1)\in\bigcap_{j=0}^4 \im N_j$, where~$N_j$ are as in\eqnref{e-Niexample}.

\section{Dualizing Trellises}\label{S-Dual}
\setcounter{equation}{0}

In this section we will investigate two methods of dualizing a given trellis in order to
obtain a trellis representing the dual code.
One method is what we call local dualization, and it amounts to taking duals of the transition spaces along with a sign inverter.
This very elegant and profound method has been introduced by Forney in~\cite{Fo01} and applies to all linear (even group) realizations;
see also the presentation based on graphical models in~\cite{Fo11}.
Theorem~\ref{T-LD} below is simply a special case of it.
The second method comes naturally with the BCJR-construction and has been introduced by Nori/Shankar in~\cite{NoSh06}.
We will show that, in general, the BCJR-dual is a proper subtrellis of the local dual and that for KV-trellises the two coincide.
There is yet another notion of trellis duality, introduced by Koetter/Vardy in
\cite[Ch.~VII]{KoVa03}, based on what they call the intersection product.
As it turns out via straightforward computation, this notion is identical to the local dual.

We begin with the local dualization, which in our particular case of tail-biting trellises looks as follows.

\begin{theo}\label{T-LD}
Let $T=(V,E)$ be a linear trellis representing the code $\cC\subseteq\F^n$.
Let $\hat{V}_j,\,j\in\cI$, be vector spaces such that $\dim V_j=\dim\hat{V}_j$ for all $j\in\cI$, and fix
non-degenerate bilinear forms $\inner{\,\cdot\,,\,\cdot\,}$ on $V_j\times\hat{V}_j,\,j\in\cI$.
For each transition space $E_j\subseteq V_j\times\F\times V_{j+1}$, define $(E_j)^{\circ}$
as the dual space with respect to the bilinear form
\[
    (V_j\times\F\times V_{j+1})\times(\hat{V}_j\times\F\times\hat{V}_{j+1})\longrightarrow \F,\quad
      \big((v,a,w),\,(\hat{v},b,\hat{w})\big)\longmapsto \inner{v,\hat{v}}+ab-\inner{w,\hat{w}},
\]
that is,
\begin{equation}\label{e-Eicirc}
  (E_j)^{\circ}\!\!:=\!\big\{(\hat{v},b,\hat{w})\in\hat{V}_j\times\F\times\hat{V}_{j+1}\,\big|\,
       \inner{v,\hat{v}}+ab-\inner{w,\hat{w}}=0\text{ for all }(v,a,w)\in E_j\big\}.
\end{equation}
Then the trellis $T^{\circ}=(\hat{V},E^{\circ})$, where $\hat{V}=\bigcup_{j=0}^{n-1}\hat{V}_j$ and
$E^{\circ}=\bigcup_{j=0}^{n-1}(E_j)^{\circ}$, is linear and represents~$\cC^{\perp}$.
Furthermore, $\text{SCP}(T^{\circ})=\text{SCP}(T):=(s_0,\ldots,s_{n-1})$ and
\begin{equation}\label{e-dimElocal}
  \dim(E_j)^{\circ}=s_j+s_{j+1}+1-\dim E_j\text{ for }j\in\cI.
\end{equation}
We call~$T^{\circ}$ the local dual of~$T$.
\end{theo}
One should notice that, via the non-degenerate bilinear form, the space~$\hat{V}_j$ is isomorphic to the
linear algebra dual of~$V_j$ consisting of all linear functionals on~$V_j$.
Conversely, the linear algebra dual naturally gives rise to a non-degenerate bilinear form and therefore may
serve as dual state space~$\hat{V}_j$.
Later on we will make specific choices for~$\hat{V}_j$ and the bilinear form, justifying our
setting in Theorem~\ref{T-LD}.

It should be noted that the isomorphism class of the trellis~$T^{\circ}$ does not depend on the choice of the
spaces~$\hat{V}_i$ and the non-degenerate bilinear forms.
Indeed, if $\inner{\,\cdot\,,\,\cdot\,}_1$ and $\inner{\,\cdot\,,\,\cdot\,}_2$ are two such forms on~$V_j\times\hat{V}_j$
and $V_j\times\tilde{V}_j$, then there exists an isomorphism $\phi_j:\,\hat{V}_j\rightarrow\tilde{V}_j$ such that
$\inner{v,\,w}_1=\inner{v,\phi_j(w)}_2$ for all $v\in V_j,\,w\in\hat{V}_j$.
As a consequence, this isomorphism furnishes a trellis isomorphism between the two corresponding dual trellises.

The theorem is a special case of the local dualization procedure for normal graphs derived in~\cite[Sec.~VII]{Fo01}.
However, we think it is worth reproducing Forney's proof for this special case.
%

\begin{proof}
It is clear that the label code $\cS(T^{\circ})$, see\eqnref{e-ST}, as well as the transition spaces $(E_j)^{\circ}$ are linear spaces.
Hence~$T^{\circ}$ is a linear trellis.
Moreover, the bilinear form on
$(V_j\times\F\times V_{j+1})\times(\hat{V}_j\times\F\times\hat{V}_{j+1})$ is non-degenerate, and therefore
$\dim(E_j)^{\circ}=\dim(\hat{V}_j\times\F\times\hat{V}_{j+1})-\dim E_j$, proving\eqnref{e-dimElocal}.
\\
It remains to show that $T^{\circ}$ represents $\cC^{\perp}$.
In order to do so, define
\begin{align*}
  \cV&:=V_0\times V_1\times V_1\times V_2\times V_2\times\ldots\times V_{n-1}\times V_{n-1}\times V_0,\\
  \hat{\cV}&:=\hat{V}_0\times\hat{V}_1\times\hat{V}_1\times\hat{V}_2\times\hat{V}_2\times\ldots\times\hat{V}_{n-1}\times\hat{V}_{n-1}\times\hat{V}_0,
\end{align*}
(that is, we are replicating each state space).
The given bilinear forms on each~$V_j\times\hat{V}_j$ naturally give rise to a non-degenerate bilinear form $\inner{\,\cdot\,,\,\cdot\,}$
on~$\cV\times\hat{\cV}$ via
\[
  \inner{(\tilde{v}_0,v_1,\tilde{v}_1,\!\ldots,\!v_{n-1},\tilde{v}_{n-1},v_0),
         (\tilde{w}_0,w_1,\tilde{w}_1,\!\ldots,\!w_{n-1},\tilde{w}_{n-1},w_0)}
  =\!\sum_{j=0}^{n-1}\!\inner{\tilde{v}_j,\tilde{w}_j}+\sum_{j=0}^{n-1}\!\inner{v_j,w_j}.
\]
This further extends to a non-degenerate bilinear form
\begin{equation}\label{e-bif}
  (\cV\times\F^n)\times(\hat{\cV}\times\F^n)\longrightarrow\F,\quad
  \inner{(v,a),\,(w,b)}= \inner{v,w}+ab\T.
\end{equation}
Recall that the transition spaces $E_j$ are contained in $V_j\times\F\times V_{j+1}$.
Hence, the direct product $E_0\times\ldots\times E_{n-1}$ is in $\cV\times\F^n$, if we
sort the state and edge labels accordingly.
Denoting this obvious permutation by~$\rho$, we therefore have
\[
   \cP:=\rho(E_0\times E_1\times\ldots\times E_{n-1})\subseteq\cV\times\F^n.
\]
Now we use the replication space
\[
  \cR:=\{(v_0,v_1,v_1,v_2,v_2,\ldots,v_{n-1},v_{n-1},v_0)\mid v_j\in V_j\}\subseteq\cV
\]
in order to relate edges in the various transition spaces~$E_j$ with each other by checking whether
ending and starting points coincide.
Indeed, the label code $\cS(T)$ defined in\eqnref{e-ST} is
given by the conditioned space
\[
  (\cP\mid\cR):=\{(v,a)\in\cP\mid v\in\cR\}
\]
(to be precise, we also have to eliminate for~$v\in\cR$ one copy of each state to obtain~$\cS(T)$).
As a consequence, the code~$\cC$ is the projection
\[
  \cC=\pi(\cP\mid\cR)=\{a\in\F^n\mid \exists\; v\in\cV:\; (v,a)\in(\cP\mid\cR)\},
\]
where~$\pi$ denotes the natural projection of $\cV\times\F^n$ onto~$\F^n$ (this is the identity $\cC=(\cP\mid\cR)|_{I_{\cA}}$ in
\cite[p.~540]{Fo01}).
Now we may take duals with respect to $\inner{\,\cdot\,,\,\cdot\,}$.
Then Forney's Theorem on Conditioned Code Duality
\cite[Thm.~7.2]{Fo01} (or straightforward verification) tells us that
$\cC^{\perp}=\big[\pi(\cP\mid\cR)\big]^{\perp}=\pi(\cP^{\perp}\mid\cR^{\perp})$.
Here $\cdot^{\;\perp}$ refers to the bilinear form in\eqnref{e-bif} or a restriction of it to the appropriate subspaces.
It remains to clarify the relation between the latter space and the trellis~$T^{\circ}$.
From the very definition of~$\cP$ and~$\cR$ and the bilinear forms involved we obtain
\begin{align*}
  &\cP^{\perp}=\rho(E_0^{\perp}\times\ldots\times E_{n-1}^{\perp}),\\
  &\cR^{\perp}=\{(w_0,-w_1,w_1,-w_2,w_2,\ldots,-w_{n-1},w_{n-1},-w_0)\mid w_j\in\hat{V}_j\},
\end{align*}
where
$E_j^{\perp}=\{(w,b,\tilde{w})\in\hat{V}_j\times\F\times\hat{V}_{j+1}\mid \inner{v,w}+ab+\inner{\tilde{v},\tilde{w}}=0
 \text{ for all }(v,a,\tilde{v})\in E_j\}$.
Observing that $(w_j,b_j,-w_{j+1})\in E_j^{\perp}\Longleftrightarrow (w_j,b_j,w_{j+1})\in(E_j)^{\circ}$, we see that
\[
  \pi(\cP^{\perp}\mid\cR^{\perp})=\{b\in\F^n\mid \text{there exists a cycle in~$T^{\circ}$ with edge-label sequence }b\}.
\]
This proves that $T^{\circ}$ represents $\cC^{\perp}$, as desired.
\end{proof}

The following two examples illustrate that in specific cases local dualization may lead to undesirable trellises.
While in the first example this is not surprising because the primary trellis is not even proper, the second example
is more unexpected.
It shows a BCJR-trellis that is isomorphic to the corresponding product trellis and thus non-mergeable (hence biproper)
and one-to-one, and yet the local dual is not reduced.
In Theorem~\ref{T-ProdBCJRdual} below we will see that this does not happen for KV-trellises.
It will be shown that the local dual of a KV-trellis is isomorphic to the BCJR-dual and thus reduced.
Furthermore, in Section~\ref{S-DualProc} we will show that this dual is even a KV-trellis of the dual code.

\begin{exa}\label{E-Localdual1}
\begin{arabiclist}
\item Consider the 2-dimensional code
      \[
          \cC=\im\begin{pmatrix}0&1&1\\1&0&1\end{pmatrix}\subseteq\F_2^3
      \]
      and choose the span list $\cS=[(1,2],(0,2]]$.
      Then the corresponding product trellis~$T=T_{G,\cS}$ has SCP $(0,1,2)$ and ECP $(1,2,2)$ and is shown in the figure below.
      Notice that~$T$ is a conventional trellis, but not proper (and thus not minimal).
      The transition spaces~$E_j$ of~$T$ can be read off from the matrix
      \[
        (M_0\,|\,G_0^{\sf T}\,|\,M_1\,|\,G_1^{\sf T}\,|\,M_2\,|\,G_2^{\sf T}\,|\,M_0)=\left(\!\begin{array}{cc|c|cc|c|cc|c|cc}
            0&0&0&0&0&1&1&0&1&0&0\\0&0&1&0&1&0&0&1&1&0&0\end{array}\!\right);
      \]
      see Definition~\ref{D-Prodtrellis} for the state space matrices~$M_j$ of product trellises.
      According to Theorem~\ref{T-LD}, the local dual~$T^{\circ}$ has SCP $(0,1,2)$ and ECP $(1,2,1)$.
      In order to compute~$T^{\circ}$, we observe that the standard bilinear form on $\F_2^2$ induces a
      non-degenerate form on each~$V_j=\im M_j$, and thus may be used for the computation of the dual transition spaces~$(E_j)^{\circ}$.
      In particular, we will use~$V_j$ for the dual state spaces as well.
      Then we compute
      \[
        (E_0)^{\circ}=\im\!\!\left(\!\!\begin{array}{cc|c|cc}0\!&\!0\!&\!1\!&\!0\!&\!1\end{array}\!\!\right),\
        (E_1)^{\circ}=\im\!\!\left(\!\!\begin{array}{cc|c|cc}0\!&\!0\!&\!1\!&\!1\!&\!0\\0\!&\!1\!&\!0\!&\!0\!&\!1\end{array}\!\!\right),\
        (E_2)^{\circ}=\im\!\!\left(\!\!\begin{array}{cc|c|cc}1\!&\!1\!&\!1\!&\!0\!&\!0\end{array}\!\!\right).
      \]
      This leads to the following trellis~$T^{\circ}$
      \\[1ex]
      \mbox{}\hspace*{2cm}\includegraphics[height=4cm]{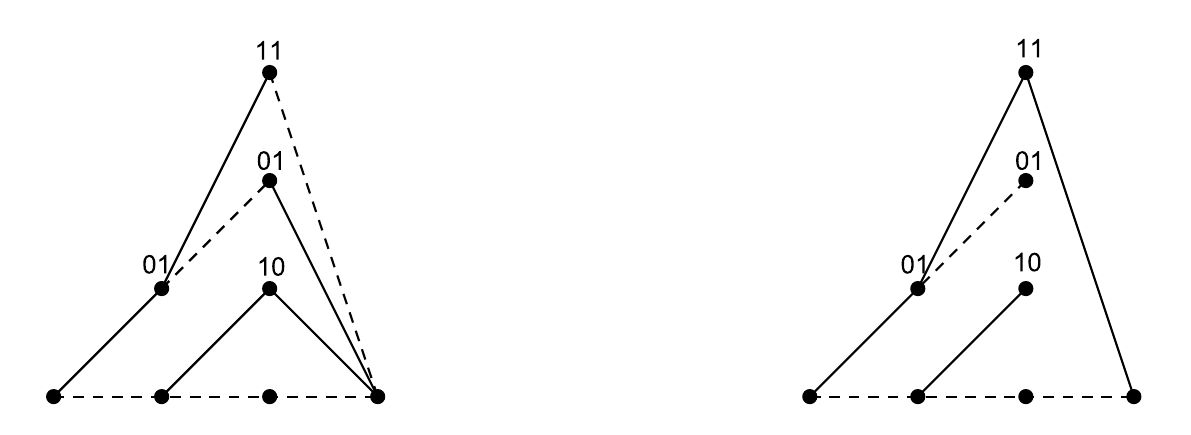}
      \\
      \mbox{}\hspace*{3.5cm}{\footnotesize [Trellis $T$]}\hspace*{5.5cm}{\footnotesize [Trellis $T^{\circ}$]}
      \\[1ex]
      Obviously, not every vertex appears in a cycle and thus the trellis~$T^{\circ}$ is not reduced.
      As a consequence,~$T^{\circ}$ is not a product trellis in the sense of Definition~\ref{D-Prodtrellis}.
\item Consider Example~\ref{E-selfdualBCJR}(a).
        The BCJR-trellis $T:=T_{(G,H,\cS)}$ given in that example has state and transition spaces
        $V_j=\im N_j$ and $E_j=\im(N_j,G_j^{\sf T},N_{j+1})$, with all matrices displayed in the matrix\eqnref{e-Niexample}.
        In order to compute the local dual~$T^{\circ}$, we may again use the standard bilinear form on~$\F_2^2$ and
        thus let~$V_j$ be the dual state space as well.
        Then
        \[
         \begin{array}[t]{l}
           (E_0)^{\circ}=\im\!\left(\!\begin{array}{cc|c|cc}1&0&1&0&0\\0&1&0&0&1\end{array}\!\right),\\[2.3ex]
           (E_1)^{\circ}=\im\!\left(\!\begin{array}{cc|c|cc}0&1&1&0&1\end{array}\!\right),\\[.7ex]
           (E_2)^{\circ}=\im\!\left(\!\begin{array}{cc|c|cc}0&0&1&1&0\\0&1&0&1&1\end{array}\!\right),\\
         \end{array}
         \quad
         \begin{array}[t]{l}
           (E_3)^{\circ}=\im\!\left(\!\begin{array}{cc|c|cc}1&0&1&1&0\\0&1&0&0&1\end{array}\!\right),\\[2.3ex]
           (E_4)^{\circ}=\im\!\left(\!\begin{array}{cc|c|cc}1&0&0&1&1\\0&1&0&0&1\\0&1&1&0&0\end{array}\!\right).
         \end{array}
        \]
        This leads to the first trellis in the following figure.
        \\[1ex]
        \mbox{}\hspace*{1cm}\includegraphics[height=3.5cm]{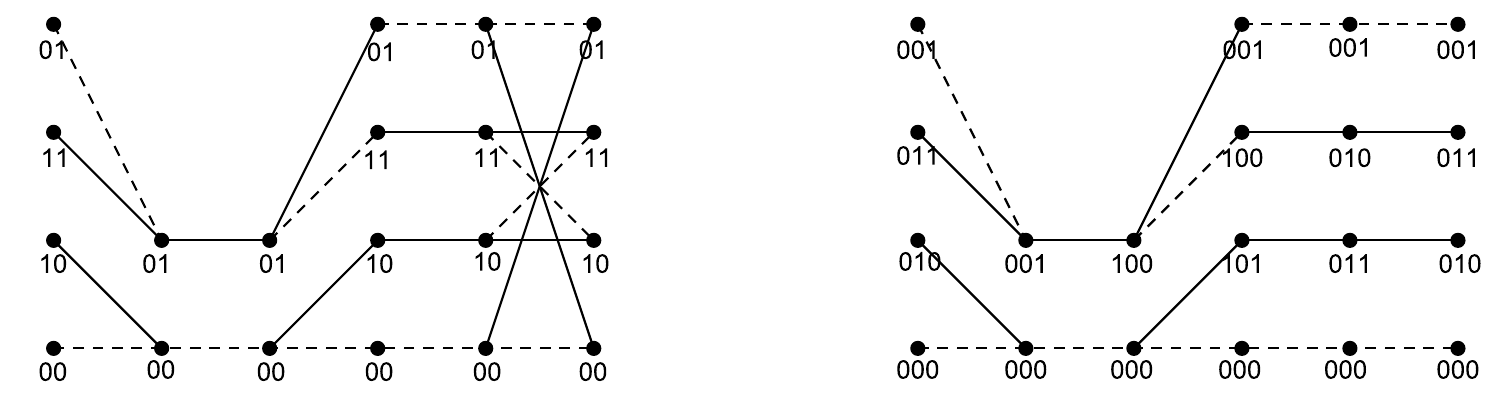}
        \\
        \mbox{}\hspace*{2.7cm}{\footnotesize [Trellis $T^{\circ}$]}\hspace*{6cm}{\footnotesize [Trellis $T^{\perp}$]}
        \\[1ex]
        The trellis~$T^{\circ}$ is not reduced because not every edge appears in a cycle.
        Indeed, the four diagonals in $(E_4)^{\circ}$, the last section of the trellis,  are not part of any cycle in~$T^{\circ}$.
        If we remove these~$4$ edges, then we obtain an isomorphic copy of the trellis on the right-hand side, which
        still represents~$\cC^{\perp}$.
        In the next result we will discuss the dualization leading to~$T^{\perp}$.
        Isomorphic versions of the trellises~$T^{\circ}$ and~$T^{\perp}$ appeared already in \cite[Ex.~IV.12, Rem.~V.4]{GLW10}.
\end{arabiclist}
\end{exa}

We now continue with a very simple and natural way of dualizing the trellises defined in Definition~\ref{D-NS}(a).
This has been introduced by Nori/Shankar in~\cite{NoSh06}.

\begin{prop}[\mbox{\cite[Def.~11]{NoSh06}}]\label{P-NSdual}
Let $T=T_{(G,H,D)}$ be as in Definition~\ref{D-NS}(a) and suppose~$T$ represents the code $\cC$.
Then the trellis $T_{(H,G,D^{\sf T})}$ represents the dual code~$\cC^{\perp}$.
We call $T_{(H,G,D^{\sf T})}$ the BCJR-dual of~the trellis~$T$, denoted by~$T^{\perp}$.
\end{prop}

One should bear in mind that even if the trellis~$T$ is a BCJR-trellis in the sense of Definition~\ref{D-NS}(b),
that is, its displacement matrix is based on a span list, then~$T^{\perp}$ is not necessarily a
BCJR-trellis in that sense, but only a trellis of the type defined in part~(a) of that definition.
This can be seen from Example~\ref{E-Localdual1}(2) above.
The trellis~$T$ given therein is a BCJR-trellis, and one can easily check that the displayed trellis~$T^{\perp}$
is indeed its BCJR-dual.
Obviously the trellis~$T^{\perp}$ is mergeable (merging the states $(000)$ and $(011)$ in $V_0$ does not create any new cycles),
and thus~$T^{\perp}$ is not a BCJR-trellis due to Theorem~\ref{T-BCJR}(a).

For trellises of the form $T=T_{(G,H,D)}$, we now have two ways of dualizing them, both of which result
in trellises representing the dual code.
By construction, the dual trellises $T^{\circ}$ and $T^{\perp}$ have the same SCP.
In general, however, these trellises are not isomorphic, as we have seen already in Example~\ref{E-Localdual1}(2) above.
Next we will show that, just like in the above example,~$T^{\perp}$ is a subtrellis of~$T^{\circ}$.

\begin{prop}\label{P-NSLD}
Let $T=T_{(G,H,D)}$ be as in Definition~\ref{D-NS}(a).
Let $\hat{E}_j$ and $(E_j)^{\circ}$ be the transition spaces of the duals~$T^{\perp}$ and $T^{\circ}$, respectively.
Then $\hat{E}_j\subseteq(E_j)^{\circ}$, up to trellis isomorphism.
\end{prop}
In the proof we will construct the local dual based on a suitable choice of dual state space and inner form, which
will then make~$T^{\perp}$ a true subtrellis of~$T^{\circ}$ and not just an isomorphic copy.

\begin{proof}
Let~$V_j=\im N_j$ and $E_j=\im(N_j,\,G_j^{\sf T},\,N_{j+1})$ be the state spaces and transition spaces of~$T$,
where the matrices $N_j$ are defined as in\eqnref{e-NMat}.
By the very definition of the BCJR-dual, the state spaces of~$T^{\perp}$ are given by $\hat{V}_j=\im\hat{N}_j$,
where $\hat{N}_j=N_j^{\sf T}$.
Notice that the bilinear form $V_j\times\hat{V}_j\longrightarrow\F$, defined as
$\inner{\alpha N_j,\beta\hat{N}_j}:=\alpha N_j\beta\T$, is well-defined and non-degenerate.
So we may construct the local dual~$T^{\circ}$ based on this form.
Obviously, $\dim\hat{V}_j=\dim V_j$ for all $j\in\cI$, and the transition spaces of~$T^{\circ}$ are
\[
  (E_j)^{\circ}=\bigg\{(\beta\hat{N}_j,b,\tilde{\beta}\hat{N}_{j+1})\in\hat{V}_j\times\F\times\hat{V}_{j+1}\,\bigg|\,
       \begin{array}{l}\alpha N_j\beta\T+\alpha G_j^{\sf T} b-\alpha N_{j+1}\tilde{\beta}\T=0\\
                    \text{for all }\alpha(N_j,G_j^{\sf T},N_{j+1})\in E_j\end{array}\bigg\}.
\]
Now we see that
$\hat{E}_j=\im(\hat{N}_j,\,H_j^{\sf T},\hat{N}_{j+1})\subseteq(E_j)^{\circ}$ since for
all $\beta(\hat{N}_j,\,H_j^{\sf T},\hat{N}_{j+1})\in\hat{E}_j$ and $\alpha\in\F^k$ we have
$\alpha N_j\beta\T+\alpha G_j^{\sf T}H_j\beta\T-\alpha N_{j+1}\beta\T=\alpha(N_j+G_j^{\sf T}H_j-N_{j+1})\beta\T=0$ , due to\eqnref{e-NMat}.
\end{proof}

Now we are ready to show the main result of this section.

\begin{theo}\label{T-ProdBCJRdual}
If $T$ is a KV-trellis of~$\cC$, then~$T^{\perp}$ is isomorphic to~$T^{\circ}$.
\end{theo}
One should keep in mind that at this point it is not clear whether $T^{\perp}$, hence~$T^{\circ}$, is a KV-trellis
of the dual code.
This is indeed the case, as we will prove in Section~\ref{S-DualProc}.

\begin{proof}
By Theorem~\ref{T-KVBCJR} we may represent~$T$ as a BCJR-trellis
$T=T_{(G,H,\cS)}$, where  $G\in\F^{k\times n}$ has rank~$k$ and $\cS=[(a_l,b_l],l=1,\ldots,k]$ is a span
list of~$G$ consisting of characteristic spans of~$\cC$.
Put $\cA=\{a_1,\ldots,a_k\}$ and $\cB=\{b_1,\ldots,b_k\}$.
Let $N_0=D$ and~$N_j$ be as in\eqnref{e-DMat},\eqnref{e-NMat}.
By Proposition~\ref{P-NSdual}, the trellis $T^{\perp}$ represents~$\cC^{\perp}$ and has
state and transition spaces $\hat{V}_j:=\im\hat{N}_j$ and $\hat{E}_j:=\im(\hat{N}_j,H_j^{\sf T},\hat{N}_{j+1})$,
where $\hat{N}_j=N_j^{\sf T}$.
In light of Proposition~\ref{P-NSLD} it suffices to show that $\dim(E_j)^{\circ}=\dim\hat{E}_j$ for all $j\in\cI$.
\\
Denote the SCP and ECP of~$T$ by $(s_0,\ldots,s_{n-1})$ and $(e_0,\ldots,e_{n-1})$, respectively,
and let $e_j^{\circ}=\dim(E_j)^{\circ}$.
Since~$T$ is isomorphic to the corresponding product trellis~$T_{G,\cS}$, the formulas in
Proposition~\ref{P-formulas}(d) apply.
Thus we have $e_j=s_{j+1}$ if $j\not\in\cB$ and $e_j=s_{j+1}+1$ if $j\in\cB$.
Using\eqnref{e-dimElocal} this leads to
\begin{equation}\label{e-ecirc}
  e_j^{\circ}=s_j+1\text{ if }j\not\in\cB\ \text{ and }\ e_j^{\circ}=s_j\text{ if }j\in\cB.
\end{equation}
On the other hand, the recursion in\eqnref{e-NMat} implies that
\begin{equation}\label{e-ehatformulas}
   \hat{e}_j:=\dim(\hat{E}_j)=\rk(\hat{N}_j,H_j^{\sf T},\hat{N}_{j+1})=\rk(\hat{N}_j,H_j^{\sf T}).
\end{equation}
Using that $\hat{N}_j=N_j^{\sf T}$ has rank~$s_j$, we conclude that $\hat{e}_j=s_j$ iff
$H_j\in\im N_j$ and $\hat{e}_j=s_j+1$ otherwise.
But then Lemma~\ref{L-HN} shows that
$\hat{e}_j=s_j$ iff $j\in\cB$ and $\hat{e}_j=s_j+1$ iff $j\not\in\cB$.
A comparison with\eqnref{e-ecirc} establishes $\dim(E_j)^{\circ}=\dim\hat{E}_j$ for all $j\in\cI$.
So the trellises are isomorphic.
\end{proof}

The following example shows that $T^{\perp}\cong T^{\circ}$  may be true even if~$T$ is
not a KV-trellis.

\begin{exa}\label{E-BCJRlocaldual}
Let $\cC=\im G=\ker H\T\subseteq\F_2^6$, where
\[
   G=\begin{pmatrix}0&1&1&1&1&1\\0&0&1&1&1&0\\1&1&0&1&0&1\end{pmatrix},\
   H=\begin{pmatrix}0&0&1&0&1&0\\1&0&0&1&1&0\\1&1&1&1&0&1\end{pmatrix}.
\]
Consider the span list $\cS=[(1,5],\,(2,4],\,(3,1]]$ for~$G$.
One can easily verify that $(1,5]$ and $(2,4]$ are characteristic spans of~$\cC$, but $(3,1]$ is not
(there exists a codeword with span $(3,0]$).
Hence the product trellis $T_{G,\cS}$ is not a KV-trellis.
By straightforwardly computing the data for $T_{G,\cS}$ and the corresponding BCJR-trellis
$T:=T_{(G,H,\cS)}$, one obtains that both trellises have SCP $(1,1,1,2,3,2)$.
Hence they are isomorphic due to Theorem~\ref{T-BCJR}(c).
Their ECP is $(1,2,2,3,3,2)$.
The displacement matrix of  $T_{(G,H,\cS)}$ is given by
\[
    N_0=\begin{pmatrix}0&0&0\\0&0&0\\0&1&0\end{pmatrix}.
\]
Let us now consider the BCJR-trellis $T_{(H,G,\hat{\cS})}$ of~$\cC^{\perp}$,
where $\hat{\cS}=[(2,4],\,(3,0],\,(0,5]]$ is the chosen span list for the rows of~$H$.
Its displacement matrix turns out to be $N_0^{\sf T}$.
As a consequence,
$T^{\perp}=T_{(G,H,\cS)}\,^{\!\!\perp}=T_{(G,H,N_0)}\,^{\!\!\perp}=T_{(H,G,N_0^{\sf T})}=T_{(H,G,\hat{\cS})}$.
One can also easily verify that $T_{(H,G,\hat{\cS})}\cong T_{H,\hat{\cS}}$.
Now it is easy to check that both~$T^{\perp}$ and~$T^{\circ}$ have ECP $(2,1,2,3,3,2)$.
Thus Proposition~\ref{P-NSLD} yields $T^{\perp}\cong T^{\circ}$.
\end{exa}

\section{A Duality for Characteristic Matrices} \label{S-DualProc}
\setcounter{equation}{0}
In this section we will restrict ourselves to KV-trellises.
Recall from Theorem~\ref{T-ProdBCJRdual} that for these trellises the local dual and the BCJR-dual coincide, and
hence we may use these dualizations interchangeably.
In~\cite[p.~2097]{KoVa03}, Koetter/Vardy formulated the conjecture that if~$(X,\cT)$ and~$(Y,\,\hat{\cT})$ are characteristic pairs
of~$\cC$ and~$\cC^{\perp}$, respectively, then for every a KV$_{(X,\cT)}$-trellis there is a
KV$_{(Y,\hat{\cT})}$-trellis with the same state complexity profile.
They actually restricted the characteristic matrices to a specific choice, the lexicographically first ones, and also made
the correspondence more precise.
In \cite[Ex.~III.13, Ex.~V.2]{GLW10}, however, it has been shown that the conjecture is not true in this generality, and in
particular not for the lexicographically first characteristic matrices of~$\cC$ and~$\cC^{\perp}$.
In this section we will prove the following reformulated, but much stronger, version of the Koetter/Vardy conjecture:
for each characteristic pair $(X,\cT)$ of~$\cC$ there exists a characteristic pair $(Y,\hat{\cT})$ of~$\cC^{\perp}$
such that the dual of a KV$_{(X,\cT)}$-trellis is a KV$_{(Y,\hat{\cT})}$-trellis.
Since, due to Theorem~\ref{T-LD}, the local dual and the primary trellis have the same SCP, this result covers indeed
Koetter/Vardy's original conjecture.
We will also explicitly construct the dual characteristic matrix~$Y$.
The details of our result also extend a theorem in~\cite[Thm.~V.3]{GLW10}, which showed how the characteristic
span list of the dual of a minimal trellis of~$\cC$ looks (recall from Theorem~\ref{T-KVresults}(b) that minimal trellises are KV-trellises).

A first step toward formulating the duality conjecture had been done in~\cite{KoVa03}.
Our general assumption that both codes~$\cC$ and~$\cC^{\perp}$ have support~$\cI$ is crucial for this result
and the rest of this section.
\begin{theo}[\mbox{\cite[Thm.~5.12]{KoVa03}}]\label{T-KVreversed}
If $\cT=[(a_l,b_l],l=1,\ldots,n]$ is the characteristic span list of~$\cC$, then the characteristic
span list of~$\cC^{\perp}$ is given by $[(b_l,a_l],l=1,\ldots,n]$.
\end{theo}

As a consequence, one easily derives from Proposition~\ref{P-formulas}(d) and Definition~\ref{D-CharMat}(iv)
that if one selects~$k$ linearly
independent characteristic generators of~$\cC$ with spans, say, $[(a_l,b_l],l=1,\ldots,k]$ and $n-k$
characteristic generators of~$\cC^{\perp}$ that do {\sl not\/} have spans $(b_l,a_l],l=1,\ldots,k$, then
the resulting product trellises have the same SCP; this has been proven in \cite[Prop.~5.13]{KoVa03}.
However, it is not guaranteed that those $n-k$ dual characteristic generators are linearly independent and
thus generate~$\cC^{\perp}$.
Indeed, this is not the case in general; see \cite[Ex.~III.13, Ex.~V.2]{GLW10}.
Furthermore, even if those dual generators are linearly independent, they may not give rise to a KV-trellis
dual to the KV-trellis of~$\cC$; see Example~\ref{E-F3manyCharMat} below.
Despite these obstacles, the result in \cite[Prop.~5.13]{KoVa03} indicates how the dual should look:
if the dual of a KV-trellis of~$\cC$ is a KV-trellis of~$\cC^{\perp}$, then its span list has to be given by the
reversed complementary spans in the above sense.
This is indeed the only option for the dual characteristic span list because in \cite[Prop.~III.14]{GLW10}
it has been shown that different selections of~$n-k$ characteristic spans lead to non-isomorphic trellises.
All of this leads to the following useful notation.

\begin{defi}\label{D-dualselection}
Let $(X,\cT)$ and $(Y,\hat{\cT})$ be characteristic matrices of the codes $\cC$ and~$\cC^{\perp}$, respectively.
A pair $(\tilde{X},\,\tilde{Y})$ is called a {\sl dual selection\/}
of~$(X,\,Y)$ if~$\tilde{X}\in\F^{k\times n}$ is a submatrix of~$X$ and~$\tilde{Y}\in\F^{(n-k)\times n}$ is a submatrix of~$Y$
such that their corresponding span lists $\cS\subset\cT$ and $\hat{\cS}\subset\hat{\cT}$ satisfy
$\hat{\cS}=\hat{\cT}\,\backslash\,[(b,a]\mid (a,b]\in\cS]$.
\end{defi}

In this section we will prove the following.
\begin{theo}\label{T-XY}
Let~$(X,\cT)$ be a characteristic pair of~$\cC$, and let~$\hat{\cT}$ be the characteristic span list of~$\cC^{\perp}$.
Then there exists a characteristic matrix~$Y$ of~$\cC^{\perp}$ such that
each dual selection $(\tilde{X},\tilde{Y})$ of $(X,Y)$ satisfies the following properties.
\begin{arabiclist}
\item $\rk\tilde{X}=k\Longleftrightarrow\rk\tilde{Y}=n-k$.
\item Let $\rk\tilde{X}=k$ and let~$\cS\subset\cT$ and~$\hat{\cS}\subset\hat{\cT}$ be the characteristic span lists
      of~$\tilde{X}$ and~$\tilde{Y}$, respectively.
      Then the KV-trellises $T_{\tilde{X},\cS}$ and $T_{\tilde{Y},\hat{\cS}}$ are dual to each other.
\end{arabiclist}
\end{theo}
Recall that in the situation~(2) the trellises $T_{\tilde{X},\cS}$ and $T_{\tilde{Y},\hat{\cS}}$ represent~$\cC$ and~$\cC^{\perp}$, respectively.
We will show this result by explicitly constructing the dual matrix~$Y$.
Here is an outline of the quite technical procedure.
Fix a characteristic matrix~$X$ of~$\cC=\ker H\T$ and consider the BCJR-trellis $T_{(X,H,\cT)}$ based on the entire matrix~$X$.
For each $m=1,\ldots,n$ consider the subtrellis~$T_m$ of~$T$ generated by all characteristic generators except
the $m$-th one.
Let the omitted generator have span $(a_m,b_m]$.
Then one can find a cycle in the local dual $(T_m)^{\circ}$ with span $(b_m,a_m]$, and hence the associated edge-label sequence
is a characteristic generator of~$\cC^{\perp}$.
This is carried out in Proposition~\ref{P-MainConstr} (with the special case where $b_m=0$ appearing in Lemma~\ref{L-DualCycle}).
Collecting all of these dual characteristic generators results in a characteristic matrix~$Y$ for~$\cC^{\perp}$.
This is the matrix~$Y$ that will satisfy Theorem~\ref{T-XY}.
Unfortunately, it is not a priori clear whether the pair $(X,Y)$ satisfies the dual rank condition in part~(1) of Theorem~\ref{T-XY}.
Its somewhat technical proof is given in Proposition~\ref{P-DualRank}.
For this step, the BCJR-representation of the trellises turns out to be crucial as it provides us
with a close link between states and dual codewords.
Once the dual rank condition is established, Theorem~\ref{T-XY}(2) is essentially a consequence of the construction.
Indeed, pick a subset $\cK\subset\cI$ such that $|\cK|=k$, and let~$\tilde{X}$ consist of the characteristic generators
in~$X$ with span list $\cS=[(a_l,b_l]\mid l\in\cK]$.
Then the KV-trellis $T_{(\tilde{X},H,\cS)}$ is a subtrellis of~$T_m$ for each $m\not\in\cK$, and therefore
each characteristic generator in~$Y$ with span in $[(b_m,a_m]\mid m\not\in\cK]$
appears in the local dual of $T_{(\tilde{X},H,\cS)}$.
As a consequence, Theorem~\ref{T-XY}(1) tells us that if $\rk\tilde{X}=k$, then these
$n-k$ generators give rise to the entire local dual of $T_{(\tilde{X},H,\cS)}$.
The only detail that needs attention is the choice of the dual state spaces and the bilinear form.
But this comes with the construction of~$Y$, and it is easy to verify that it is indeed non-degenerate,
see Proposition~\ref{P-LocalDual}.
Finally, in Theorem~\ref{T-BCJRdualKV} we will make the isomorphism, hidden in~(2) of Theorem~\ref{T-XY},
explicit by using suitable BCJR-representations.

Throughout this section we will use the following notation.
\begin{notation}\label{N-XTnotation}
Let $\cC=\ker H\T\subseteq\F^n$ be a $k$-dimensional code as in\eqnref{e-Cdata} and\eqnref{e-Gdata}.
Recall that we assume that both~$\cC$ and $\cC^{\perp}$ have support~$\cI$.
Furthermore, let $(X,\cT)$ be a fixed characteristic pair of~$\cC$, and write
$X=(x_{lj})_{l,j=0,\ldots,n-1}$, where the rows of~$X$ are sorted such that
$\cT=[(a_l,l],l=0,\ldots,n-1]$.
Hence we index the characteristic generators by $0,\ldots,n-1$.
Since~$\cC^{\perp}$ has support~$\cI$,
no characteristic span is empty and therefore $\cI\backslash(a_l,l]=(l,a_l]$ for all $l\in\cI$.
Let $X_0^{\sf T},\ldots,X_{n-1}^{\sf T}$ be the columns of~$X$ (we will always denote the $j$-th column
of a matrix~$M$ by $M_j^{\sf T}$).
Let $T=T_{(X,H,\cT)}$ be the associated BCJR-trellis.
The state space matrices of~$T$ will be denoted by~$N_j\in\F^{n\times(n-k)},\,j\in\cI$.
Recall that the BCJR-construction is completely row-wise, that is, the $l$-th rows of the matrices~$N_j$
solely depend on the generator with span $(a_l,l]$.
In the following, we will consider various submatrices of~$N_j$ and~$X$.
In all of these instances, we will index the rows of the submatrices by the ending point~$l$ of the
span $(a_l,l]$ they correspond to.
In this situation, Theorem~\ref{T-KVBCJR} reads as
\begin{arabiclist}
\item $\row(N_j,l)=0$ for all~$l$ such that $j\not\in(a_l,l]$,
\item the set $\{\!\row(N_j,l)\mid \text{ $l$ such that }j\in(a_l,l]\}$ is linearly independent for all $j\in\cI$.
\end{arabiclist}
\end{notation}

We start with the following technical lemma.
It results in a characteristic generator $c\in\cC^{\perp}$ with characteristic span starting at time~$0$
and which is normalized in such a way that~$c_0=1$.

\begin{lemma}\label{L-DualCycle}
Let $G\in\F^{(n-1)\times n}$ be the matrix consisting of all except the first row of~$X$ and let
$\cS=[(a_l,l],l=1,\ldots,n-1]$ be its span list.
Consider the BCJR-trellis $T_{(G,H,\cS)}$ with state space matrices $Q_j\in\F^{(n-1)\times(n-k)},\,j\in\cI$.
\\
Then there exists a unique vector $v\in\F^{n-k}$ such that $G_0^{\sf T}=Q_1 v\T$.
For $j\in\cI$ put
\[
    w_j=\left\{\begin{array}{ll}v,&\text{if }j\in(0,a_0],\\ 0,&\text{otherwise}\end{array}\right.
\]
and define the dual codeword $c=(c_0,\ldots,c_{n-1}):=v H\in\cC^{\perp}$.
Then
\begin{romanlist}
\item the dual codeword~$c$ has span $(0,a_0]$ and $c_0=1$. 
\item $Q_j v\T=0$ for all $j\not\in(0,a_0]$.
\item $Q_jw_j^{\sf T}+G_j^{\sf T}c_j-Q_{j+1}w_{j+1}^{\sf T}=0$ for all $j\in\cI$.
\end{romanlist}
\end{lemma}
Observe that due to Definition~\ref{D-CharMat}(iv), there exist~$k$ conventional characteristic spans
(that is, spans not containing~$0$).
Since $0\in(a_0,0]$, this implies that~$G$ contains all~$k$ characteristic generators with conventional spans.
In other words,~$G$ contains an MSGM of~$\cC$.
Therefore $\im G=\cC$, and the trellis $T_{(G,H,\cS)}$ does indeed represent~$\cC$.

\begin{proof}
By definition, $G=(x_{lj})_{l=1,\ldots, n-1\atop j=0,\ldots,n-1}$.
First notice that statement~(iii) follows from the previous parts due to the identities
$c_j=v H_j^{\sf T}=H_jv\T$ and $Q_jw_j^{\sf T}=Q_jv\T$ along with the recursion
\begin{equation}\label{e-recurQ}
   Q_{j+1}=Q_j+G_j^{\sf T}H_j \text{ for }j\in\cI,
\end{equation}
see\eqnref{e-NMat}.
Thus it remains to show the existence of~$v$ and properties~(i) and~(ii).
Let us first collect some properties of the matrices~$Q_j$.
By definition of BCJR-trellises, they are the submatrices of~$N_j$ obtained by omitting the first row, where~$N_j$ is as
in Notation~\ref{N-XTnotation}.
Due to property~(iv) of Definition~\ref{D-CharMat}, each index $j\in\cI$ is contained in
exactly $n-k$ of the characteristic spans of~$\cC$.
With the aid of Notation~\ref{N-XTnotation}(1) and~(2), we may therefore conclude that
\begin{equation}\label{e-rankQi}
  \rk Q_j=\left\{\begin{array}{ll} n-k,&\text{if }j\in(0,a_0],\\
              n-k-1,&\text{if }j\not\in(0,a_0].\end{array}\right.
\end{equation}
Moreover,
\begin{equation}\label{e-QiGiQi+1}
   \rk(Q_j,G_j^{\sf T})=\rk(G_j^{\sf T},Q_{j+1})=\rk(Q_j,G_j^{\sf T},Q_{j+1})
   =\left\{\!\!\begin{array}{ll}
      \rk Q_j,&\text{if }j=a_0,\\ \rk Q_j+1,&\text{if }j\not=a_0,\\
      \rk Q_{j+1},\!\!&\text{if }j=0,\\ \rk Q_{j+1}+1,&\text{if }j\not=0.
    \end{array}\right.
\end{equation}
The first two identities are a consequence of the recursion\eqnref{e-recurQ}.
As for the~$4$ cases of the last identity, recall that by Theorem~\ref{T-KVBCJR},
the trellis $T_{(G,H,\cS)}$ is isomorphic to the product trellis $T_{G,\cS}$.
Therefore, we may apply Proposition~\ref{P-formulas}(d).
Since the set of starting points is $\cI\backslash\{a_0\}$ and the set of ending points is
$\cI\backslash\{0\}$, we immediately obtain the 4 cases above.
\\
Now\eqnref{e-QiGiQi+1} and\eqnref{e-rankQi} imply the existence of a unique $v\in\F^{n-k}$ such that
$G_0^{\sf T}=Q_1 v\T$.
This results in $G_0^{\sf T}=(Q_0+G_0^{\sf T}H_0)v\T$, and thus $G_0^{\sf T}(1-H_0v\T)=Q_0v\T$.
But the second case in\eqnref{e-QiGiQi+1} shows that~$G_0^{\sf T}$ is not in the column space of~$Q_0$, and
therefore we may conclude that $c_0=H_0v\T=1$ and $Q_0v\T=0$.
If $a_0=n-1$, then all of this shows that the span of~$c$ is contained in $(0,n-1]$.
Since Theorem~\ref{T-KVreversed} gives that the latter is a characteristic span of~$\cC^{\perp}$, Lemma~\ref{L-ShortestGen}
implies that the span of~$c$ equals $(0,n-1]$. This proves~(i) and~(ii) for $a_0=n-1$.
Let us now assume $a_0<n-1$.
From the very definition of BCJR-trellises we know that
\[
    \row(Q_0,l)=\sum_{j=a_l}^{n-1}x_{lj}H_j,\ l=1,\ldots,n-1;
\]
for the row indexing see Notation~\ref{N-XTnotation}.
Recall that $\{a_1,\ldots,a_{n-1}\}=\{0,\ldots,n-1\}\backslash\{a_0\}$.
Hence there exists some~$l$ for which $a_l=n-1$.
Since $x_{l,a_l}\not=0$, the identity $\row(Q_0,l)v\T=0$ yields $H_{n-1}v\T=0$.
Now we can proceed recursively to obtain $c_j=H_jv\T=0$ for all $j=n-1,n-2,\ldots,a_0+1$.
Hence the span of~$c$ is contained in $(0,a_0]$, and again Lemma~\ref{L-ShortestGen} yields equality.
This establishes~(i).
\\
As for~(ii) we also proceed backwards.
From $Q_0v\T=0$ and $H_{n-1}v\T=0$ we obtain $Q_{n-1}v\T=(Q_0-G_{n-1}^{\sf T}H_{n-1})v\T=0$
and successively $Q_jv\T=0$ for $j=n-1,n-2,\ldots,a_0+1$, as desired.
This completes the proof.
\end{proof}

Using the behavior of the characteristic matrix and BCJR-trellises under the cyclic shift as summarized
in Remark~\ref{R-shift}, the lemma above generalizes to any characteristic span.
This leads to the following result.

\begin{prop}\label{P-MainConstr}
Fix $m\in\cI$ and denote by $X^m\in\F^{(n-1)\times n}$ and $N_j^m\in\F^{(n-1)\times(n-k)}$
the submatrices of~$X$ and~$N_j$, respectively, where the $m$-th row has been omitted.
Denote the columns of~$X^m$ by $(X^m)_0^{\sf T},\ldots,(X^m)_{n-1}^{\sf T}$.
Then there exists a unique vector $v_m\in\F^{n-k}$ with the following properties:
\begin{alphalist}
\item $(X^m)_m^{\sf T}=N_{m+1}^mv_m^{\sf T}$,
\item the dual codeword $c^m=(c^m_0,\ldots,c^m_{n-1}):=v_mH\in\cC^{\perp}$ has span $(m,a_m]$ and satisfies $c^m_m=1$,
\item $N_j^mv_m^{\sf T}=0$ for $j\not\in(m,a_m]$,
\item $N_j^mw_{m,j}^{\sf T}+(X^m)_j^{\sf T} c_j^m-N^m_{j+1}w_{m,j+1}^{\sf T}=0$ for all $j\in\cI$,
      where $w_{m,j}\in\F^{n-k}$ is defined as
      \[   \label{e-wi}
          w_{m,j}=\left\{\begin{array}{ll}v_m,&\text{if }j\in(m,a_m],\\0,&\text{otherwise }.\end{array}\right.
      \]
\end{alphalist}
As a consequence, the pair $(Y,\hat{\cT})$, where
\begin{equation}\label{e-YhatT}
    Y=\begin{pmatrix}c^0\\ \vdots\\c^{n-1}\end{pmatrix},\
    \hat{\cT}=[(m,a_m],m=0,\ldots,n-1],
\end{equation}
is a characteristic pair of~$\cC^{\perp}$.
\end{prop}
\begin{proof}
For $m=0$ this is exactly the statement of Lemma~\ref{L-DualCycle}:
Part~(a) translates into $G_0^{\sf T}=Q_1v\T$ in that lemma and (b),~(c),~(d) are~(i),~(ii),~(iii) of that lemma, respectively.
For $m>0$ we may apply the left cyclic shift~$\sigma^m$ by~$m$ units on the characteristic generators.
As detailed in Remark~\ref{R-shift}, this results in a characteristic
pair $(X^*,\cT^*)$ for the code $\sigma^m(\cC)=\im G^*=\ker H^*\T$ with BCJR-trellis $T_{(X^*,H^*,\cT^*)}$.
Its state space matrices are given by $N_j^*=N_{j+m}$ for $j\in\cI$.
This shifts the span $(a_m,m]$ to the span $(a_m-m,0]$ and we may use Lemma~\ref{L-DualCycle} again.
Applying now the inverse shift~$\sigma^{n-m}$ leads to the statements in~(a) --~(d).
The consequence about $(Y,\hat{\cT})$ is clear due to Theorem~\ref{T-KVreversed}.
\end{proof}

One should observe that the identities in Proposition~\ref{P-MainConstr}(d) indicate that the
dual codeword~$c^m$ and the states $w_{m,j},\,j\in\cI$, give rise to a cycle in the local dual of the trellis
$T_{(X^m,H,\cT^m)}$, where~$\cT^m$ is the characteristic span list of~$X^m$; see also\eqnref{e-Eicirc}.
We will make this precise later on after specifying the dual state spaces and the bilinear form as needed
for the local dualization.

\begin{exa}\label{E-selfdual2}
Consider the self-dual code from Example~\ref{E-selfdualBCJR}(b) with characteristic matrix~$X$ given
in\eqnref{e-charmatSD} and characteristic span list $\cT=[(3,0],(2,1],(1,2],(0,3]]$ ordered as
required in Notation~\ref{N-XTnotation}.
The matrix, displaying all information about the trellis~$T_{(X,G,\cT)}$, is given by
\[
  S=(N_0|X_0^{\sf T}|N_1|X_1^{\sf T}|N_2|X_2^{\sf T}|N_3|X_3^{\sf T}|N_0)=\left(\!\!\begin{array}{cc|c|cc|c|cc|c|cc|c|cc}
      1&0&1&0&0&0&0&0&0&0&0&1&1&0\\1&1&0&1&1&1&0&0&1&1&1&0&1&1\\
      0&0&0&0&0&1&1&1&1&0&0&0&0&0\\0&0&1&1&0&1&0&1&1&1&0&1&0&0\end{array}\!\!\right).
\]
Omitting the $m$-th row from~$S$, where $m=0,\ldots,3$, we see that
\[
  (X^0)_0^{\sf T}\!=\!\begin{pmatrix}0\\0\\1\end{pmatrix}\!=\!N^0_1\begin{pmatrix}1\\1\end{pmatrix},\:
  (X^1)_1^{\sf T}\!=\!\begin{pmatrix}0\\1\\1\end{pmatrix}\!=\!N^1_2\begin{pmatrix}0\\1\end{pmatrix},\:
  (X^2)_2^{\sf T}\!=\!\begin{pmatrix}0\\1\\1\end{pmatrix}\!=\!N^2_3\begin{pmatrix}1\\0\end{pmatrix},\:
  (X^3)_3^{\sf T}\!=\!\begin{pmatrix}1\\0\\0\end{pmatrix}\!=\!N^3_0\begin{pmatrix}1\\1\end{pmatrix}.
\]
Hence $v_0=(1\,1),\;v_1=(0\,1),\;v_2=(1\,0)$, and $v_3=(1\,1)$.
Recalling that $H=G$ due to self-duality, we obtain
$c^0=v_0G=(1001),\,c^1=v_1G=(0110),\;c^2=v_2G=(1111)$, and $c^3=v_3G=(1001)$.
Thus, the characteristic pair~$(Y,\hat{\cT})$ for~$\cC^{\perp}=\cC$, resulting from this procedure, is given by
\[
   Y=\begin{pmatrix}c^0\\c^1\\c^2\\c^3\end{pmatrix}=\begin{pmatrix}1&0&0&1\\0&1&1&0\\1&1&1&1\\1&0&0&1\end{pmatrix},\
   \hat{\cT}=[(0,3],(1,2],(2,1],(3,0]].
\]
One should notice that even though the code~$\cC$ is self-dual, the matrix~$Y$ is not identical to the characteristic
matrix~$X$ of~$\cC$ in\eqnref{e-charmatSD}, with which we started the procedure!
\\
As mentioned earlier, the identity in Proposition~\ref{P-MainConstr}(d) gives rise to cycles in certain local duals, namely
$(w_{m,0},c^m_0,w_{m,1},c^m_1,w_{m,2},c^m_2,w_{m,3},c^m_3,w_{m,0})$.
The matrix
\[
  \hat{S}=\left(\!\!\begin{array}{cc|c|cc|c|cc|c|cc|c|cc}
            0&0&1&  1&1&0&  1&1&0&  1&1&1& 0&0\\
            0&0&0&  0&0&1&  0&1&1&  0&0&0& 0&0\\
            1&0&1&  1&0&1&  0&0&1&  1&0&1& 1&0\\
            1&1&1&  0&0&0&  0&0&0&  0&0&1& 1&1\end{array}\!\!\right)
\]
contains these four cycles.
Proposition~\ref{P-MainConstr}(d) now states that for $m=0,\ldots,3$, the $m$-th row of~$\hat{S}$
is in the local dual of the trellis determined by the three rows of~$S$ not having index~$m$.
Using the definition of the local dual in\eqnref{e-Eicirc} this can also  be verified directly.
\end{exa}

Now we can formulate the main result of this section.

\begin{theo}\label{T-XYY}
Let $(X,\cT)$ be as in Notation~\ref{N-XTnotation} and~$Y$ as in Proposition~\ref{P-MainConstr}.
Then the pair $(X,Y)$ satisfies Theorem~\ref{T-XY}.
\end{theo}

The proof consists of several steps. We begin with Part~(1) of Theorem~\ref{T-XY}.
It is worth pointing out that for the proof of this result the BCJR-description of the
trellises involved is crucial.
It directly links the ``dual states''~$v_m$ to the associated dual codewords~$c^m$, making the second equivalence
in\eqnref{e-RankEqu} below obvious.

\begin{prop}\label{P-DualRank}
Let~$v_m$ and $(Y,\hat{\cT})$ be as in Proposition~\ref{P-MainConstr}, and let $\cK\cupdot\hat{\cK}=\cI$ be
an index partition such that $|\cK|=k$.
Let~$\tilde{X}\in\F^{k\times n}$ be the submatrix of~$X$ consisting of the generators with spans in the list
$\cS:=[(a_l,l]\mid l\in\cK]$, while~$\tilde{Y}\in\F^{(n-k)\times n}$ is the submatrix of~$Y$ consisting of the dual
generators with spans in $\hat{\cS}:=[(m,a_m]\mid m\in\hat{\cK}]$.
Then $(\tilde{X},\tilde{Y})$ is a dual selection of $(X,Y)$ in the sense of Definition~\ref{D-dualselection} and
\begin{equation}\label{e-RankEqu}
   \rk\tilde{X}=k\Longleftrightarrow\{v_m\mid m\in\hat{\cK}\}\text{ is linearly independent }\Longleftrightarrow\rk\tilde{Y}=n-k.
\end{equation}
\end{prop}
\begin{proof}
By Proposition~\ref{P-MainConstr}(b) we have $\tilde{Y}=\big(v_mH\big)_{m\in\hat{\cK}}$.
Using that $\rk H=n-k$, this immediately establishes the second equivalence in\eqnref{e-RankEqu}.
\\
As for the first equivalence, we need some preparation.
Recall the matrixes~$N_j$ from Notation~\ref{N-XTnotation}, and define $\tilde{N}_j\in\F^{k\times(n-k)}$ as
the submatrices of~$N_j$ consisting of the rows with indices in~$\cK$.
Consider also the matrices~$X^m$ and~$N_j^m$ as defined in Proposition~\ref{P-MainConstr}.
Since~$\cK$ and~$\hat{\cK}$ are disjoint, we observe that for each $m\in\hat{\cK}$ the
matrices~$\tilde{X}$ and~$\tilde{N}_j$ are the submatrices of~$X^m$ and~$N_j^m$, respectively,
consisting of the rows with indices in~$\cK$.
Denote the columns of~$\tilde{X}$ by $ \tilde{X}_0^{\sf T},\ldots, \tilde{X}_{n-1}^{\sf T}$.
Then Proposition~\ref{P-MainConstr}(a) yields
\begin{equation}\label{e-coltildeX}
  \tilde{X}_m^{\sf T}=\tilde{N}_{m+1}v_m^{\sf T}\text{ for all }m\in\hat{\cK}.
\end{equation}
Now we can prove the equivalence.
\\
``$\Longleftarrow$'' First suppose that $\rk\tilde{X}<k$, and let
$\alpha\tilde{X}=0$ for some $\alpha\in\F^k\backslash\{0\}$.
We will index the vector $\alpha$ by $l\in\cK$ to be consistent with our indexing of the rows of~$\tilde{X}$.
With the aid of the recursion $\tilde{N}_{j+1}=\tilde{N}_j+\tilde{X}_j^{\sf T}H_j$, we compute
$\alpha\tilde{N}_{j+1}=\alpha(\tilde{N}_j+\tilde{X}_j^{\sf T}H_j)=\alpha\tilde{N}_{j}$ for all $j\in\cI$.
Thus, $\alpha\tilde{N}_0=\alpha\tilde{N}_1=\ldots=\alpha\tilde{N}_{n-1}=:w$.
Suppose $w=0$, and hence $\alpha\in\bigcap_{j=0}^{n-1}\ker\tilde{N}_j$.
But then Notation~\ref{N-XTnotation}(2) tells us that the identity
$\alpha\tilde{N}_j=0$ leads to $\alpha_l=0$ for each $l\in\cK$ such that $j\in(a_l,l]$.
Since no span $(a_l,l]$ is empty, this results in $\alpha=0$, contradicting our assumption.
Hence $w\not=0$.
Next, along with\eqnref{e-coltildeX}, the identity $\alpha\tilde{X}=0$ implies $\alpha\tilde{N}_{m+1}v_m^{\sf T}=0$
and hence $w v_m^{\sf T}=0$ for all $m\in\hat{\cK}$.
This shows that the matrix $(v_m^{\sf T},\,m\in\hat{\cK})\in\F^{(n-k)\times(n-k)}$ is singular, and
therefore the set $\{v_m\mid m\in\hat{\cK}\}$ is linearly dependent.
\\
``$\Longrightarrow$''
Assume $\rk\tilde{X}=k$.
Then  $\cC=\im\tilde{X}=\ker H\T$, and the matrices~$\tilde{N}_j$ are
the state space matrices of the BCJR-trellis $T_{(\tilde{X},H,\cS)}$.
Let $w(v_m^{\sf T},\,m\in\hat{\cK})=0$ for some $w\in\F^{n-k}$.
We first show that $w\in\im N_j^m$ for all $m\in\hat{\cK}$ and all $j\in\cI$.
In order to do so, notice that the rank condition in\eqnref{e-rankQi} generalizes to~$N_j^m$ as
$\rk N_j^m= n-k$ if $j\in(m,a_m]$ and $\rk N_j^m= n-k-1$ if $j\not\in(m,a_m]$.
Thus, if $j\in(m,a_m]$, then $\im N_j^m=\F^{n-k}$, and hence $w\in\im N_j^m$.
If $j\not\in(m,a_m]$, then $\rk N_j^m=n-k-1$ and Proposition~\ref{P-MainConstr}(c) yields
$\im N_j^m=(\im v_m)^{\perp}$ in $\F^{n-k}$.
Thus, $w\in\im N_j^m$, as desired.
\\
Now we have $w=\alpha_j^m N_j^m$ for all $m\in\hat{\cK}$ and some $\alpha_j^m\in\F^{n-1}$.
Then Notation~\ref{N-XTnotation}(1)  tells us that for every $m\in\hat{\cK}$ we have
$w\in\text{span}\{\row(N_j,l)\mid l\text{ such that }j\in(a_l,l],\,l\not=m\}$.
From this and the linear independence of the nonzero rows in~$N_j$, see
Notation~\ref{N-XTnotation}(2), we conclude that
$w\in\text{span}\{\row(N_j,l)\mid l\not\in\hat{\cK}\text{ and }j\in(a_l,l]\}$.
But the latter space is $\im\tilde{N}_j$, and, since $j\in\cI$ was arbitrary, this proves that
$w\in\bigcap_{j=0}^{n-1}\im\tilde{N}_j$.
Using that $T_{(\tilde{X},H,\cS)}$ is a KV-trellis,  we may apply Theorem~\ref{T-bigcapkerNi} and conclude that $w=0$.
As a consequence, $\{v_m\mid m\in\hat{\cK}\}$ is linearly independent.
This concludes the proof of the remaining equivalence in\eqnref{e-RankEqu}.
\end{proof}

So far we have proven the dual rank property in Theorem~\ref{T-XY}(1).
It is worth noting that this does not guarantee that part~(2) of that theorem is also satisfied.
This is due to the fact that KV-trellises based on the same span selection but on different generators need not
be isomorphic.
Indeed, we have the following example.
\begin{exa}\label{E-F3manyCharMat}
Consider again the code $\cC\subseteq\F_3^4$ from Example~\ref{E-selfdualBCJR}(c) with the characteristic matrix~$X$
as given there.
The dual code has characteristic span list $\hat{\cT}=[(1,0],\,(0,1],\,(3,2],$ $(2,3]]$ and
the matrices
\[
   Y_1=\begin{pmatrix}1&1&1&2\\1&1&0&0\\0&0&2&1\\0&0&1&2\end{pmatrix},\
   Y_2=\begin{pmatrix}1&1&2&1\\1&1&0&0\\0&0&2&1\\0&0&1&2\end{pmatrix}
\]
are both characteristic matrices of~$\cC^{\perp}$.
Both are normalized, that is, the generators have coordinate~$1$ at the
starting point of their span.
Proposition~\ref{P-MainConstr}, applied to the characteristic matrix~$X$ of~$\cC$, produces the matrix~$Y_1$.
It is easy to see that both pairs $(X,\,Y_i),\,i=1,2$, satisfy the dual rank condition~(1) of Theorem~\ref{T-XY} for all dual selections.
It is a bit more tedious to show that only~$Y_1$ satisfies part~(2) of Theorem~\ref{T-XY} for all full rank dual selections.
Indeed, the two KV-trellises resulting from the first two rows of~$Y_1$ and~$Y_2$ are not
isomorphic\footnote{In the first trellis, the unique cycle representing the codeword $(1112)\in\cC^{\perp}$
passes through the zero state at time~$1$, but this is not the case in the second trellis.},
and only the trellis resulting from~$Y_1$ is the dual of the KV-trellis of~$\cC$ corresponding to the last two rows of~$X$.
\end{exa}

Thus, it remains to prove~(2) of Theorem~\ref{T-XY} in order to complete the proof of Theorem~\ref{T-XYY}.
Precisely, if $\rk\tilde{X}=k$ in Proposition~\ref{P-DualRank}, then the pairs $(\tilde{X},\cS)$ and $(\tilde{Y},\hat{\cS})$
give rise to KV-trellises of~$\cC$ and~$\cC^{\perp}$, respectively, and we aim to show that these trellises are
duals of each other.
Since for KV-trellises, the BCJR-dual and the local dual are isomorphic, we may use either approach.
The initial idea in Proposition~\ref{P-MainConstr} was motivated by local dualization, and therefore we
will use that construction.
While we do have a description of the state spaces~$V_j$ of $T_{(\tilde{X},H,\cS)}$, we do not yet
have dual state spaces $\hat{V}_j$ and a non-degenerate bilinear form on $V_j\times\hat{V}_j$ which are needed to
establish a local duality as in Theorem~\ref{T-LD}.
Example~\ref{E-selfdual2} shows that the standard bilinear form on~$\F^{n-k}$ restricted to~$V_j$
is in general degenerate, and therefore $\hat{V}_j=V_j$ along with this bilinear form is not an option for the local dualization:
take, for instance, $\cK=\{0,1\}$; then $V_1=\im(1,1)$.
On the other hand, the example and the construction preceding it indicate what the dual state spaces should be:
the spaces generated by the ``dual states'' $v_m$.
All this leads to the following result.

\begin{prop}\label{P-LocalDual}
Let $\cK,\,\hat{\cK},\,(\tilde{X},\,\cS)$, and $(\tilde{Y},\,\hat{\cS})$ be as in Proposition~\ref{P-DualRank}.
Assume $\rk\tilde{X}=k$, and let $V_j:=\im \tilde{N}_j$, where $\tilde{N}_j\in\F^{k\times(n-k)}$ are the
state space matrices of the KV-trellis $T:=T_{(\tilde{X},H,\cS)}$.
For $j\in\cI$ define the matrices $P_j\in\F^{(n-k)\times(n-k)}$ via
\begin{equation}\label{e-Pi}
   \row(P_j,m)=w_{m,j}=\left\{\begin{array}{ll}v_m,&\text{if }j\in(m,a_m]\\0,&\text{otherwise}\end{array}\right\}
   \text{ for }m\in\hat{\cK}.
\end{equation}
Put $\hat{V}_j:=\im P_j$ and $\hat{E}_j=\im(P_j,\tilde{Y}_j^{\sf T},P_{j+1})$ for $j\in\cI$, where~$\tilde{Y}_j^{\sf T}$ denotes the
$j$-th column of $\tilde{Y}$.
Then
\begin{arabiclist}
\item $\rk\tilde{N}_j=\rk P_j=\rk\tilde{N}_jP_j^{\sf T}$ for all $j\in\cI$.
      As a consequence, $\dim V_j=\dim\hat{V}_j$ for all~$j\in\cI$ and the bilinear form
      \begin{equation}\label{e-nondegen}
         V_j\times\hat{V}_j\longrightarrow\F,\quad (\alpha\tilde{N}_j,\beta P_j)\longmapsto \alpha\tilde{N}_jP_j^{\sf T}\beta\T
      \end{equation}
      is non-degenerate.
\item The trellis~$\hat{T}=(\hat{V},\hat{E})$, where $\hat{V}=\bigcup_{j=0}^{n-1}\hat{V}_j$ and
      $\hat{E}=\bigcup_{j=0}^{n-1}\hat{E}_j$, is a linear trellis representing~$\cC^{\perp}$ and $\hat{T}=T^{\circ}$,
      where we dualize with respect to the bilinear form\eqnref{e-nondegen}.
      Moreover, the trellis~$\hat{T}$ is isomorphic to the KV$_{(Y,\hat{\cal T})}$-trellis $T_{\tilde{Y},\hat{\cS}}$.
\end{arabiclist}
\end{prop}
\begin{proof}
(1) Recall that by property~(iv) of Definition~\ref{D-CharMat} every index~$j$ is contained in exactly~$n-k$
characteristic spans.
Thus, with the aid of Notation~\ref{N-XTnotation}(1) and~(2), and $|\hat{\cK}|=n-k$ we compute
\begin{align*}
 \rk\tilde{N}_j&=|\{l\in\cK\mid j\in(a_l,l]\}|=n-k-|\{m\in\hat{\cK}\mid j\in(a_m,m]\}|\\
               &=|\{m\in\hat{\cK}\mid j\not\in(a_m,m]\}|=|\{m\in\hat{\cK}\mid j\in(m,a_m]\}|.
\end{align*}
But the last quantity is exactly~$\rk P_j$ due to Proposition~\ref{P-DualRank}.
It remains to show $\rk\tilde{N}_j=\rk\tilde{N}_jP_j^{\sf T}$.
We will do this by showing $\ker\tilde{N}_jP_j^{\sf T}=\ker\tilde{N}_j$.
Clearly, we have ``$\supseteq$''.
For the converse, let $\alpha\tilde{N}_jP_j^{\sf T}=0$ for some $\alpha\in\F^k$.
Then $\alpha\tilde{N}_jv_m^{\sf T}=0$ for all $m\in\hat{\cK}$ such that $j\in(m,a_m]$.
Along with Proposition~\ref{P-MainConstr}(c) this leads to $\alpha\tilde{N}_jv_m^{\sf T}=0$ for all
$m\in\hat{\cK}$.
But due to Proposition~\ref{P-DualRank} the matrix $(v_m\mid m\in\hat{\cK})\in\F^{(n-k)\times(n-k)}$ is
non-singular and thus we conclude $\alpha\tilde{N}_j=0$.
This concludes the proof of the rank identities.
The non-degeneracy of the bilinear form follows from
$\ker\tilde{N}_jP_j^{\sf T}=\ker\tilde{N}_j$ and $\ker P_j\tilde{N}_j\T=\ker P_j$.
\\
(2)
We first show $\hat{T}=T^{\circ}$.
In order to do so, we have to prove $\hat{E}_j=(\tilde{E}_j)^{\circ}$, where
$\tilde{E}_j=\im(\tilde{N}_j,\tilde{X}_j^{\sf T},\tilde{N}_{j+1})$ are the transition spaces of~$T$, and~$(\tilde{E}_j)^{\circ}$
are the duals in the sense of\eqnref{e-Eicirc}.
The rows of $(P_j,\tilde{Y}_j^{\sf T},P_{j+1})$ are given by $(w_{m,j},c^m_j,w_{m,j+1})$ for $m\in\hat{\cK}$.
Since $(\tilde{N}_j,\tilde{X}_j^{\sf T},\tilde{N}_{j+1})$ is a submatrix of $(N^m_j,(X^m)_j^{\sf T},N^m_{j+1})$ for all $m\in\hat{\cK}$,
Proposition~\ref{P-MainConstr}(d) shows that $\hat{E}_j\subseteq(\tilde{E}_j)^{\circ}$.
For the converse we will show that the two spaces have the same dimension.
Let $s_j:=\rk\tilde{N}_j=\rk P_j$ and $e_j=\dim\tilde{E}_j$.
Then Theorem~\ref{T-LD} tells us that $e_j^{\circ}:=\dim(\tilde{E}_j)^{\circ}=s_j+s_{j+1}+1-e_j$.
Moreover, being a KV-trellis,~$T$ is isomorphic to the product trellis $T_{\tilde{X},\cS}$.
Therefore, Proposition~\ref{P-formulas}(d) tells us
that $e_j=s_{j+1}$ if $j\in\hat{\cK}$ and $e_j=s_{j+1}+1$ if $j\in\cK$.
Thus
\[
  e_j^{\circ}=s_j+1 \text{ if }j\in\hat{\cK}\;\text{ and }\;e_j^{\circ}=s_j \text{ if }j\in\cK.
\]
Using that $\dim\hat{E}_j=\rk(P_j,\tilde{Y}_j^{\sf T},P_{j+1})\geq\rk P_j=s_j$ for all~$j$,
we obtain $\dim\hat{E}_j\geq e_j^{\circ}$ for $j\in\cK$.
\\
If $j\not\in\cK$ then $j\in\hat{\cK}$ and  $\row(P_j,j)=0$ due to\eqnref{e-Pi}.
Since the $j$-th entry of~$\tilde{Y}_j^{\sf T}$ is given by $c^j_j=1$, this yields
$\dim\hat{E}_j=\rk(P_j,\tilde{Y}_j^{\sf T},P_{j+1})\geq\rk P_j+1=s_j+1=e_j^{\circ}$.
Along with $\hat{E}_j\subseteq(\tilde{E}_j)^{\circ}$, all of this shows $\hat{E}_j=(\tilde{E}_j)^{\circ}$, and
therefore $\hat{T}=T^{\circ}$.
Consequently,~$\hat{T}$ represents~$\cC^{\perp}$.
\\
It remains to show that~$\hat{T}$ is isomorphic to the KV$_{(Y,\hat{\cal T})}$-trellis $T_{\tilde{Y},\hat{\cS}}$.
For this remember that the latter has state spaces $\im M_j$ and transition spaces
$\im(M_j,\,\tilde{Y}_j^{\sf T},\,M_{j+1})$, where~$M_j$ is the state space matrix as in Definition~\ref{D-Prodtrellis}
based on the span list $[(m,a_m]\mid m\in\hat{\cK}]$.
As a consequence, the rows of~$M_j$ with index~$m$ such that $j\in(m,a_m]$ are linearly independent
while all other rows are zero.
Comparing this with\eqnref{e-Pi} and making use of Proposition~\ref{P-DualRank}, we see that
$\alpha M_j\mapsto\alpha P_j$ induces a well-defined isomorphism between $\im M_j$ and $\hat{V}_j=\im P_j$.
Now the very definitions of the transition spaces of~$\hat{T}$ and $T_{\tilde{Y},\hat{\cS}}$ show that
this gives rise to a trellis isomorphism.
This concludes the proof.
\end{proof}

This establishes the proof of Theorem~\ref{T-XYY}.

Having the dual pairing $(X,Y)$ of characteristic matrices allows us to express the duality of their
KV-trellises in terms of BCJR-representations.
Indeed, let $(\tilde{X},\tilde{Y})$ be a dual selection of $(X,Y)$ with span lists~$\cS$ and~$\hat{\cS}$ as in
Proposition~\ref{P-DualRank} and such that $\rk\tilde{X}=k$, thus $\rk\tilde{Y}=n-k$.
Then~$\tilde{Y}$ and~$\tilde{X}$ are parity check matrices of~$\cC$ and~$\cC^{\perp}$, respectively, and
may be used for the BCJR-representations of the codes.
This results in the KV-trellises $T_{(\tilde{X},\tilde{Y},\cS)}$ of~$\cC$ and
$T_{(\tilde{Y},\tilde{X},\hat{\cS})}$ of~$\cC^{\perp}$.
By definition of the dual matrix~$Y$, we have $\tilde{Y}=UH$, where $U=(v_m)_{m\in\hat{\cK}}$.
Moreover,~$U$ is non-singular due to Proposition~\ref{P-DualRank}.
Thus, the trellis $T_{(\tilde{X},\tilde{Y},\cS)}$ is isomorphic to the trellis
$T_{(\tilde{X},H,\cS)}$; see also \cite[Rem.~IV.4]{GLW10}.
As a consequence, the trellis~$\hat{T}$ of Proposition~\ref{P-LocalDual}(2), being the local
dual of $T_{(\tilde{X},H,\cS)}$, is isomorphic to the local dual of $T_{(\tilde{X},\tilde{Y},\cS)}$, which in turn
is isomorphic to the BCJR-dual $T_{(\tilde{X},\tilde{Y},\cS)}^{\perp}$, due to Theorem~\ref{T-ProdBCJRdual}.
A tedious, but straightforward matrix computation shows that this trellis satisfies
the following symmetry.

\begin{theo}\label{T-BCJRdualKV}
Let the data be as in Proposition~\ref{P-DualRank} and let $\rk\tilde{X}=k$. Then
\[
    T_{(\tilde{X},\tilde{Y},\cS)}^{\perp}=T_{(\tilde{Y},\tilde{X},\hat{\cS})}.
\]
\end{theo}
The situation described in Example~\ref{E-F3manyCharMat} may be used to show that this symmetry is not true
for full rank dual selections of arbitrary pairs $(X,Y)$ of characteristic matrices, even if they
satisfy the dual rank condition.

Let us close the paper with the following remarks pertaining to specific classes of codes.
First, if~$\cC\subseteq\F^n$ is a cyclic code, then it is easy to see that~$\cC$
has only one characteristic matrix (up to scalar factors).
It is given by the~$n$ cyclic shifts of the generator polynomial, see also \cite[Lem.~2]{KaSh05}.
As a consequence, this matrix and the corresponding characteristic matrix for the dual code
satisfy Theorem~\ref{T-XY}.
The dual rank condition of this theorem has already been proven in \cite[Thm.~6]{KaSh05}.

Second, let us return to self-dual codes.
Notice that for the self-dual code in Example~\ref{E-selfdualBCJR}(b), the pair $(X,X)$,
where~$X$ is as in\eqnref{e-charmatSD}, does not satisfy
the dual rank condition of Theorem~\ref{T-XY}(1).
It is easy to check that~$X$ is not the lexicographically first characteristic matrix of that code
(the lexicographically first characteristic matrix is obtained by choosing for each span the
lexicographically first codeword having that span, and where the lexicographic ordering starts at the starting point of that span).
In the proof of~\cite[Thm.~1]{KaSh05} it was stated that for any self-dual code, the pair $(X,X)$, where~$X$ is the lexicographically
first characteristic matrix, satisfies the dual rank condition of Theorem~\ref{T-XY}(1).
While this is indeed the case for the particular code in Example~\ref{E-selfdualBCJR}(b),
this is, unfortunately, not true in general.


\begin{exa}\label{E-selfdualbig}
Consider the extended $[8,4,4]$-Hamming code $\cC\subseteq\F_2^8$ generated by
\[
   G=\begin{pmatrix}1&0&1&0&1&1&0&0\\0&1&1&1&1&0&0&0\\0&0&1&0&1&0&1&1\\0&0&0&1&1&1&1&0\end{pmatrix}.
\]
Then
\[
  X= \begin{pmatrix}1&0&0&0&0&1&1&1\\1&1&0&0&1&0&1&0\\1&1&1&0&0&0&0&1\\1&0&1&1&0&0&1&0\\
               0&1&1&1&1&0&0&0\\1&0&1&0&1&1&0&0\\0&0&0&1&1&1&1&0\\0&0&1&0&1&0&1&1
        \end{pmatrix},\,
  \cT=[(5,0],(4,1],(7,2],(6,3],(1,4],(0,5],(3,6],(2,7]]
\]
form a characteristic pair, and~$X$ is the lexicographically first characteristic matrix of~$\cC$.
It is easy to see that the pair $(X,X)$ does not satisfy the dual rank condition in Theorem~\ref{T-XY}(1):
the rows with spans $(4,1],\,(7,2],\,(6,3],\,(3,6]$ are linearly independent, whereas this is not
the case for the rows with spans $(5,0],\,(4,1],\,(7,2],\,(0,5]$ (these are the spans that are not the
reversed spans of the first list).
\end{exa}

\bibliographystyle{abbrv}

\end{document}